\documentclass[twocolumn]{article}

\usepackage{multirow}
\usepackage{graphicx}
\usepackage{color}
\usepackage{algorithm}
\usepackage{algorithmic}
\usepackage{subfigure}
\usepackage{verbatim}
\usepackage{array}
\usepackage{amsmath}
\usepackage{amssymb}
\usepackage{amsthm}
\usepackage{url}
\usepackage{fullpage}
\usepackage{times}

\newtheorem{obs}{Observation}
\newtheorem{theorem}{Theorem}
\newtheorem{proposition}{Proposition}
\newtheorem{lemma}{Lemma}
\newtheorem{definition}{Definition}

\begin{document}

\title{Minimum Spanning Tree on Spatio-Temporal Networks}

\author{
\hspace*{-11mm}
\begin{tabular}{ccc}
	Viswanath Gunturi & Shashi Shekhar & Arnab Bhattacharya \\
	\url{vgvm@iitk.ac.in} & \url{shekhar@cs.umn.edu} & \url{arnabb@iitk.ac.in} \\
	{\small Dept. of Computer Science and Engineering,} & {\small Dept. of
	Computer Science and Engineering,} & {\small Dept. of Computer Science and
	Engineering,}\\
	{\small Indian Institute of Technology, Kanpur} & {\small University of Minnesota, Twin Cities} & {\small Indian Institute of Technology, Kanpur}\\
	{\small Kanpur, UP 208016, India.} & {\small Minneapolis, MN 55454, USA.} & {\small Kanpur, UP 208016, India.} \\
\end{tabular}
}

\date{}

\maketitle 

\begin{abstract}
	Given a spatio-temporal network (ST network) where edge properties vary with
	time, a \emph{time-sub-interval minimum spanning tree} (TSMST) is a
	collection of minimum spanning trees of the ST network, where each tree is
	associated with a time interval. During this time interval, the total cost
	of tree is least among all the spanning trees. The TSMST problem aims to
	identify a collection of distinct minimum spanning trees and their
	respective time-sub-intervals under the constraint that the edge weight
	functions are piecewise linear. This is an important problem in ST network
	application domains such as wireless sensor networks (e.g., energy efficient
	routing). Computing TSMST is challenging because the ranking of candidate
	spanning trees is non-stationary over a given time interval. Existing
	methods such as dynamic graph algorithms and kinetic data structures assume
	separable edge weight functions. In contrast, we propose novel algorithms to
	find TSMST for large ST networks by accounting for both separable and
	non-separable piecewise linear edge weight functions.  The algorithms are
	based on the ordering of edges in edge-order-intervals and intersection
	points of edge weight functions. 
\end{abstract}

\section{Introduction}
A spatio-temporal network is a network consisting of nodes with location information and edges connecting these nodes.  
The topology and network parameters (such as edge cost) of a spatio-temporal network change with time. 
These kind of networks appear in numerous applications such as energy-efficient routing in a wireless sensor network \cite{wsn-app1,wsn-app2}.

In a typical wireless sensor network, large number of sensor networks are scattered in the observation field. Once deployed, physical 
access to these nodes may prove to be difficult. Thus, energy conservation becomes a important constraint in designing 
algorithms for these kind of sensor networks. In many wireless sensor network applications the nodes are not stationary, i.e., 
they change their physical position with time, for example, sensor network among robots on a reconnaissance mission \cite{msn2}. Usually in such scenarios,
the sensor nodes physically move on a predetermined trajectory to collect data in observation field \cite{msn1}. One of the important problems 
pertaining to these kind of sensor networks is to maintain network connectivity among the individual nodes such that the total cost of 
transmission among the nodes is minimum.  This can be modeled as a spatio-temporal network with nodes representing sensor nodes 
and transmission link among any two sensor nodes represented as an edge. The weight of an edge in the network represents 
energy required to transmit along that particular transmission link. Now the problem is to maintain an energy-efficient communication 
network among the sensor nodes such that there is a path between any two nodes and the total sum of the energy 
required for transmission for all edges involved is minimum. Traditionally minimum spanning trees have been used to solve 
these kind of problems \cite{wsn-app1,wsn-app2} in a static environment; but these no longer 
hold in a non-static environment.

Figure~\ref{wsn-a} shows a wireless sensor network maintained among a group of sensors moving on predetermined trajectories
as depicted in Figure~\ref{wsn-b}. In Figure~\ref{wsn-a}, we assume that there is no direct connectivity between the sensor 
nodes 1 and 4, and that sensor node 5 is connected only to sensor node 3. This network of sensors is represented as a ST 
network in Figure~\ref{mdt-a} where the sensors are represented as nodes and the communication link between any two sensors
is represented as an edge. Time dependent edge weights represent cost of packet transmission between sensors. Since the sensor 
nodes are moving, the distance between any two nodes changes with time.  The energy required to transmit data from one node
to another is directly proportional to the square of distance between them ~\cite{wsn-app1}. Thus, even a small change in 
distance would affect the cost of transmission significantly.  The solution to the time-sub-interval minimum spanning tree (TSMST) 
problem effectively determines the energy-efficient communication paths among these sensor nodes.\\

\begin{figure}[t]
\begin{center}
\subfigure[Sensor Network]{\label{wsn-a} \includegraphics[width=22mm]{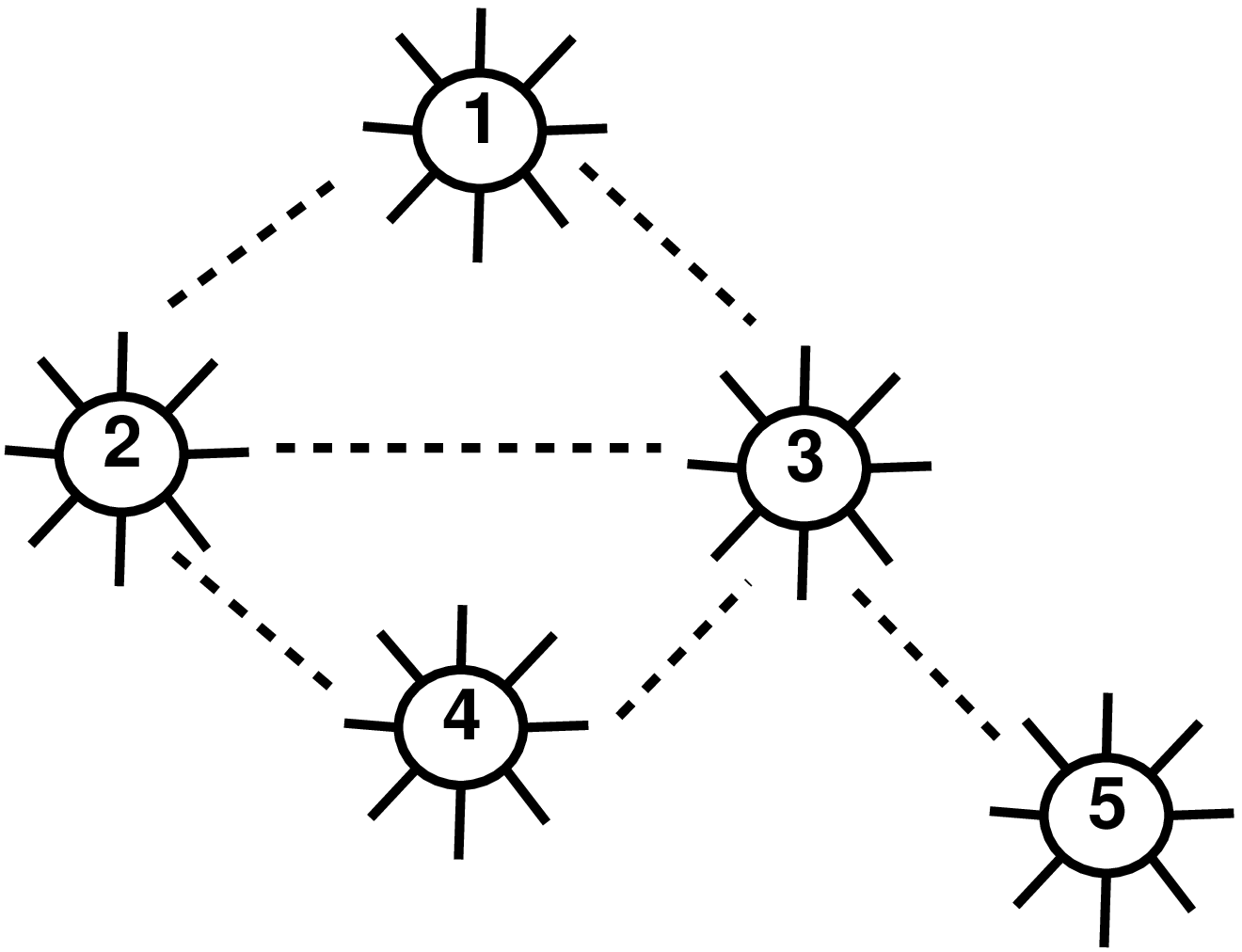}}
\subfigure[Positions of sensor nodes at various times]{\label{wsn-b} \includegraphics[width=50mm]{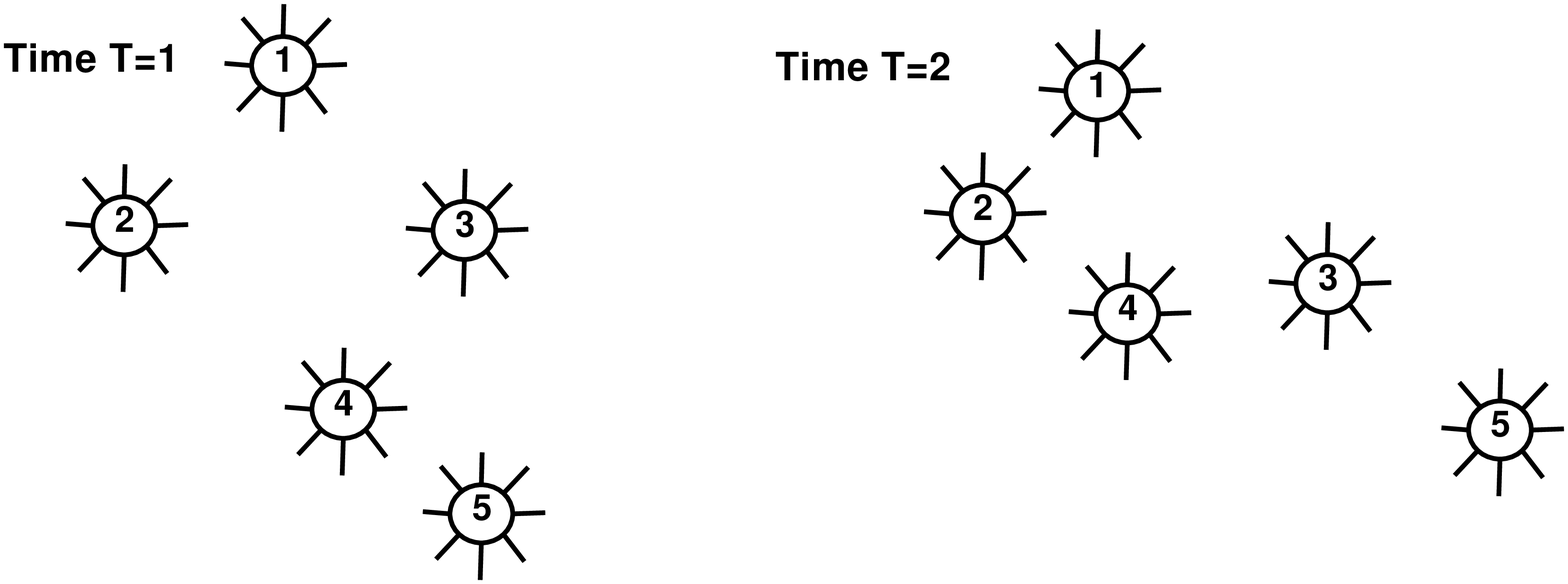}}
\caption{Sample wireless sensor network.}
\end{center}
\end{figure}

\subsection{Motivation}
The limited energy of the sensor nodes requires efficient transmission of information so
that the network lifetime is increased. A lot of work done in this area \cite{wsn-app3,wsn-app1,wsn-app2,wsn-app4,wsn-app5} assume 
that the sensors nodes are stationary.

Computing TSMST is expensive because of the non-stationary ranking of candidate spanning trees in a ST network. This is 
illustrated in Figure~\ref{mdt}. Figure~\ref{mdt-b} shows the minimum spanning trees (MSTs) at different time instants
for the ST network shown in Figure~\ref{mdt-a}. The total costs of these minimum spanning trees are given in Table~\ref{mc}.
Figure~\ref{mdt-b} shows that both spanning tree and total cost of spanning tree change with time.\\

\begin{figure*}[t]
\begin{center}
\subfigure[ST network ]{\label{mdt-a} \includegraphics[width=40mm]{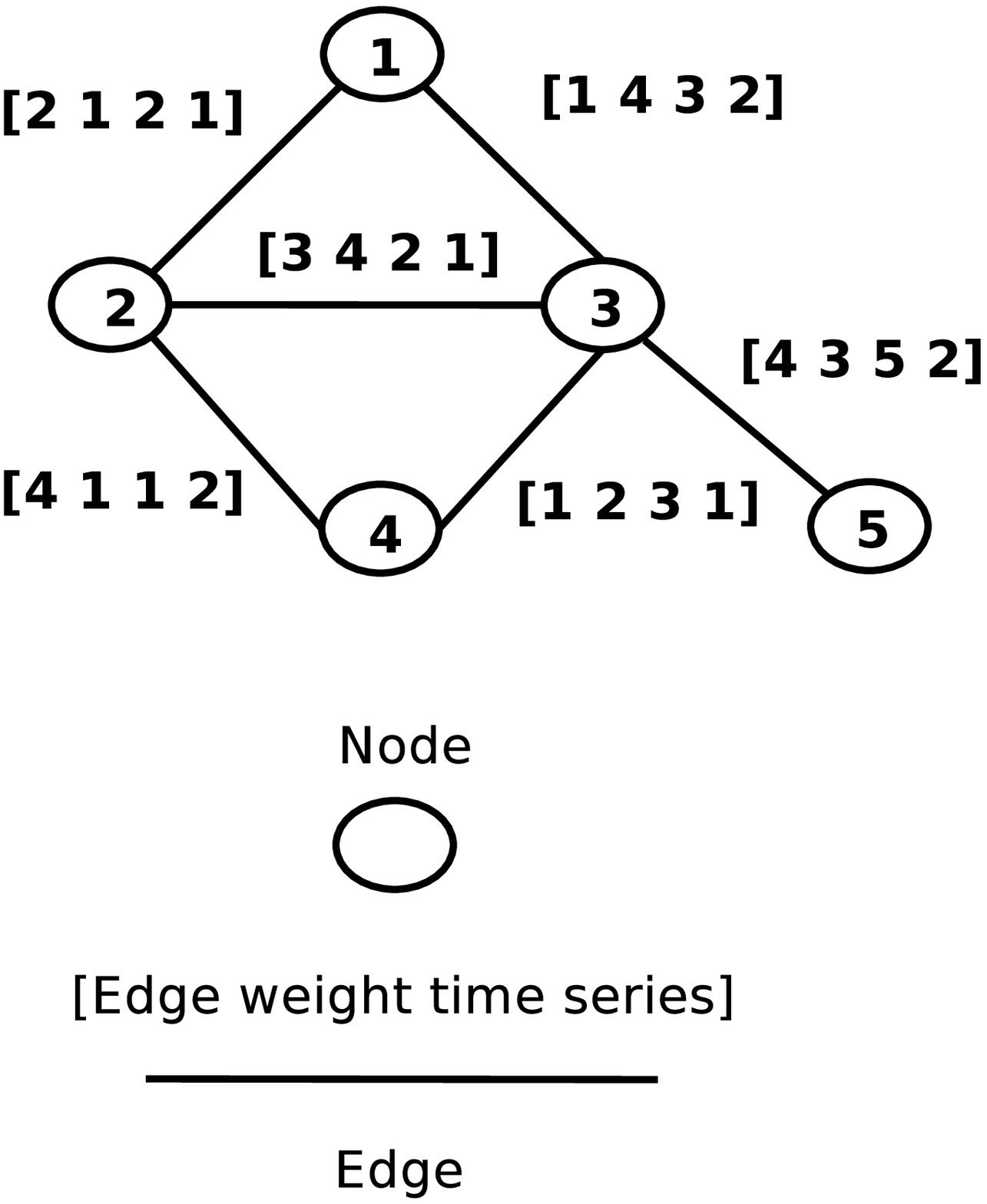}}
\subfigure[MSTs at different time instants]{\label{mdt-b}\includegraphics[width=100mm]{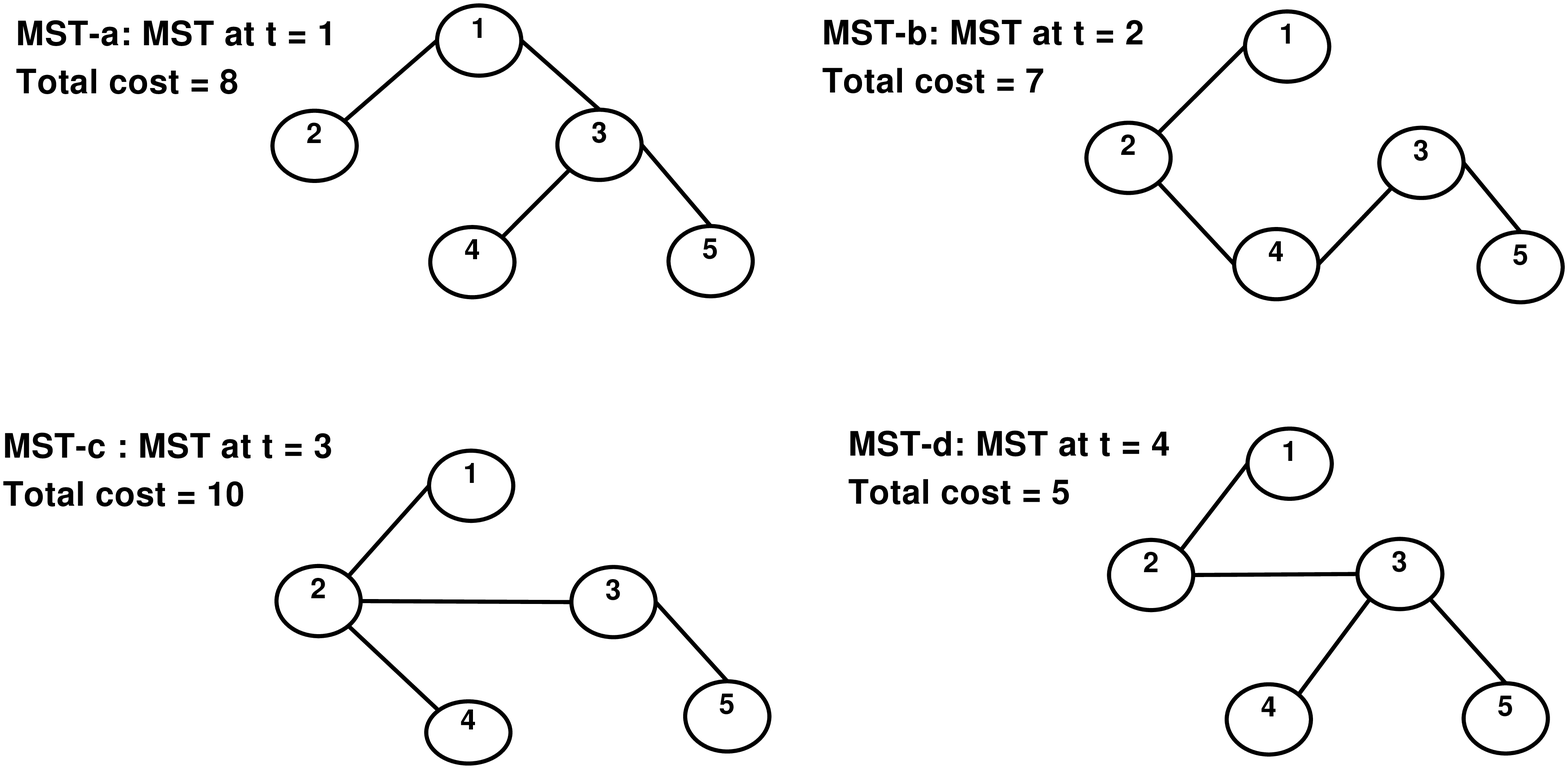}}
\caption{Spatio-temporal network and its corresponding MSTs at various times.}
\label{mdt}
\end{center}
\end{figure*}

\subsection{Contributions}
We present the problem of time-sub-interval minimum spanning tree (TSMST) in a spatio-temporal network. We propose two
algorithms to find the TSMST on a spatio-temporal network and provide analytical evaluations of the proposed algorithms. 
The algorithms allow the ST network to have both separable as well as non-separable edge weight functions.
We also present the experimental analysis of the algorithms proposed.

The rest of the thesis is organized as follows. Chapter \ref{pd} defines the concepts used in the paper, followed by the 
problem definition of TSMST on a spatio-temporal network. Related work is described in Chapter \ref{rw}. 
In Chapter \ref{tca}, two algorithms for solving the TSMST problem are presented. We present correctness and time complexity analysis
of our algorithms in Chapter \ref{eval}. Chapter \ref{exp} presents the experimental design and performance 
analysis. We conclude in Chapter \ref{conclude}.\\

\begin{table}[t]
\begin{center}
\begin{tabular}{|c|c|c|c|c|}
\hline
Time & \multicolumn{4}{c|}{Total cost}\\
\cline{2-5}
 &MST-a & MST-b & MST-c & MST-d \\
\hline
1 & 8 & 11 & 13 & 10 \\
2 & 10 & 7 & 9 & 10 \\
3 & 13 & 11 & 10 & 12 \\
4 & 5 & 6 & 6 & 5 \\
\hline
\end{tabular}
\caption{Total cost of MSTs at various times.}
\label{mc}
\end{center}
\end{table}

\section{Basic Concepts and Problem Definition}
\label{pd}
We model a spatio-temporal network as a \emph{time-aggregated-graph} (TAG)
\cite{TAGSSTD,TAG-BASE}. A time aggregated graph is a graph in which each edge
is associated with a edge weight function. These functions are defined over a
time horizon and are represented as a time series. 
For instance edge (3,5) of the graph shown in Figure \ref{mdt-a} has been assigned a time series [4 3 5 2], i.e., the
weight of the edge (3,5) at time instants t=1, 2, 3, and 4  are 4, 3, 5, and 2 respectively. The edge weight is assumed to vary linearly between any two time instants.
We also assume that no two edge weight functions have same values for two or more consecutive time instants of their time series. If such a case occurs then the values 
of any one of the edges are increased (or decreased) by small quantity $\epsilon$ to make them distinct. For example, in Figure \ref{mdt-a} weight functions of the edges 
(1,2) and (2,3) have same values for time t=3 and t=4. The edge weight functions of graph in Figure \ref{mdt-a} are shown in Figure \ref{rconvex-b}.\\

\begin{figure}[t]
\begin{center}
\includegraphics[angle=270, width= \columnwidth]{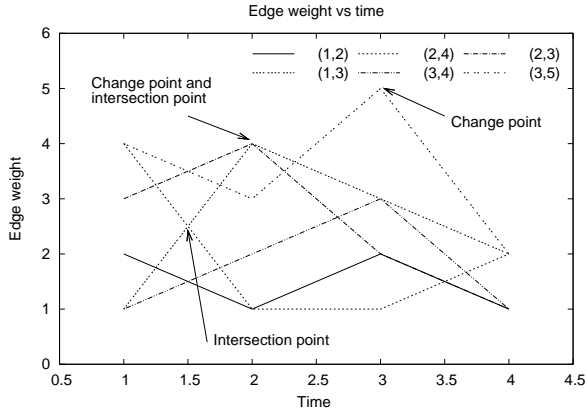}
\caption{Edge weight function plot}
\label{rconvex-b}
\end{center}
\end{figure}

\begin{definition}[Time-sub-interval]
	A time-sub-interval, denoted as $i=(i_s , i_e)$, is a maximal sub interval of time horizon $[1,K]$ which has a unique MST. This unique MST is denoted as $TMST(i)$. 
In other words, the ranking of candidate spanning trees (based on the total cost of tree) is stationary during a time-sub-interval.   
\end{definition}

For example, Figure \ref{rtrace} shows the 4 time-sub-intervals and the MSTs
during those time periods.

\begin{definition}[Edge-order-interval]
	An edge-order-interval, denoted as $\omega=(\omega_{s},\omega_{e})$, is a sub interval of time
	horizon $[1,K]$ during which there is a clear ordering of edge weight
	functions, i.e., none of them intersect with each other. 
\end{definition}

An edge-order-interval is guaranteed to have a unique MST (see Proposition \ref{lem5}).
Two or more consecutive edge-order-intervals may have the same MST. A
time-sub-interval is usually composed of one or more edge-order-intervals. For
example, in Figure \ref{rconvex-b}, the interval (2.66 3.0) is an
edge-order-interval whereas the interval (2.66 3.66] is a time-sub-interval
which is a union of three consecutive edge-order-intervals (2.66 3.0), (3.0 3.5)
and (3.5 3.66).\\

\begin{lemma}
If all the edge weights of a graph are distinct, then there is a unique minimum spanning tree.
\label{lem4} 
\end{lemma}

\begin{proof}
(By contradiction) Let $T_1$ and $T_2$ be two different minimum spanning trees of a graph. Let $OT_1$ = $e_1 e_2 e_3 e_4 \ldots e_i e_{i+1} \ldots e_{n-1}$ be the increasing order of the 
edge weights of $T_1$ and $OT_2$ = $e_1' e_2' e_3' e_4' \ldots e_i' e_{i+1}' \ldots e_{n-1}'$ be the increasing order of the edge weights of $T_2$. Here $n$ is the number of 
vertices in the graph. Without loss of generality assume that $e_j$ is same as $e_j'$, $\forall j \leq i$. Further assume that $e_{i+1} < e_{i+1}'$. Now consider the cycle generated by 
adding the edge $e_{i+1}$ to $T_2$. Only one such cycle can be created as $T_2$ is a spanning tree of the graph. Now if this cycle contains only the edges $e_j'$ where $j \leq i$, then it 
implies that there is a cycle in $T_1$ as $e_j$ is same as $e_j'$, $\forall j \leq i$ (a contradiction).  Now, if the cycle contains some edges $e_j'$ where $j \leq i$ and $e_{i+1}'$, then 
we can replace the edge $e_{i+1}'$ with $e_{i+1}$ and make a minimum spanning tree of lower cost contradicting the fact that $T_2$ is a MST. Consider the case when the cycle contains the 
edges $e_j'$ where $j \geq i+1$. Let $e_j'$ where $j \neq i+1$ be any edge of the cycle, then we have $e_{i+1} < e_{i+1}' < e_j'$. Again, we can replace edge $e_{i+1}'$ with $e_{i+1}$ and
make a MST of lower cost contradicting the fact that $T_2$ is a MST. Therefore, we can conclude that $T_1$ and $T_2$ are same.    
\end{proof}

\begin{proposition}
An edge-order-interval has a unique MST.
\label{lem5} 
\end{proposition}

\begin{proof}
Using Lemma \ref{lem4}, if all the edge weights are distinct, then there is a unique MST. In an edge-order-interval there is a clear ordering 
among the edges for all time instants (no two edge weight functions intersect); therefore, there will be an unique MST. 
\end{proof}

\begin{figure}[t]
\begin{center}
\includegraphics[width=\columnwidth]{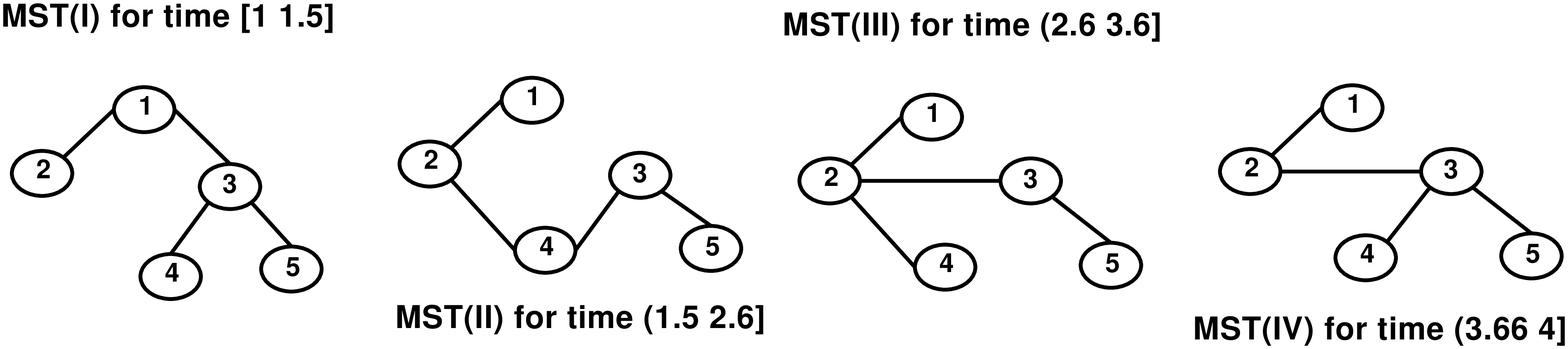}
\caption{Time sub-interval minimum spanning trees.}
\label{rtrace}
\end{center}
\end{figure}

\subsection{Problem Definition}

Given an undirected ST network $G=(V,E)$ where $V$ is the set of vertices of
graph, $E$ is the set of edges, and each edge $e \in E$ has a weight function
associated with it. The weight function is defined over the time horizon
$[1,K]$. The problem of $TSMST$ is to determine the set of distinct minimum
spanning trees, $TMST(i)$, and their respective time-sub-intervals, $i=(i_s
,i_e)$.

The total cost of $TMST(i)$ is least among all other spanning trees over its
respective time-sub-interval $i = [i_s,i_e]$ where $0\leq i_s \leq i_e \leq K$.

We assume that for all edges $e \in E$, the edge weight function is defined for
the entire time interval $[1,K]$. The weight of an edge is assumed to vary linearly between any two time instants
of time series.

In our example of an energy efficient communication network maintained by a
group of sensors, the communication network is represented as a ST network shown
in Figure \ref{mdt-a}. The collection of distinct minimum spanning trees and
their corresponding time-sub-intervals is shown in Figure \ref{rtrace}.

\section{Related Work}
\label{rw}

The related work is classified based on candidate spanning tree ranking. The traditional greedy algorithms 
for determining MST assume stationary ranking. Related work done in the area of non-stationary ranking 
assume separable edge cost. In contrast our work proposes algorithms for non-stationary ranking case by 
accounting for both separable and non-separable edge weights.

\begin{figure}[h]
\begin{center}
\includegraphics[width=60mm]{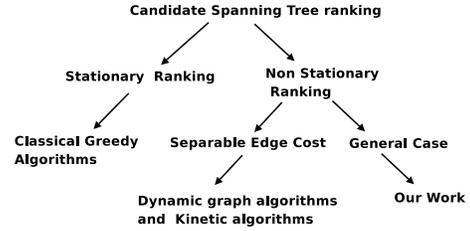}
\caption{Related work classification}
\label{dual}
\end{center}
\end{figure}

\subsection{Classical Greedy Algorithms}
Classical methods for computing minimum spanning trees~\cite{SMST,SMST1,CPROP}
were developed for static networks and assume stationary ranking of candidate
trees (see Figure \ref{dual}), i.e., they assume that mutual ranking (on basis
of total cost) among the spanning trees does not change with time. Therefore,
the classical greedy algorithms such as Kruskal's~\cite{CPROP} and Prim's~\cite{CPROP} cannot be applied
to TSMST problem on spatio-temporal networks.

\subsection{Dynamic Graph Algorithms}
Dynamic graph algorithms and kinetic algorithms incorporate non-stationary
candidate ranking by making use of dynamic data structures such as topology
trees \cite{TPT,ADS} and dynamic trees \cite{DynDS}. However, dynamic data
structures can model only discrete changes such as single edge insertion,
deletion or weight modification \cite{SPAR,TPT,DG1,DG2}. Hence, they cannot
handle piecewise linear edge weight functions.

Kinetic algorithms \cite{KMST,KA2,KA1} combine parametric optimization along
with dynamic data structures. In a parametric optimization
problem~\cite{param1}, each edge $e$ is associated with a linear weight
function. Kinetic algorithms transform the edge weight functions to a dual
plane, i.e., any edge weight function $w_e = a\lambda + b$ where $\lambda$
represents time, is transformed to a point $(-a,b)$ in the dual plane. Thus any
intersection between the edge functions of a tree edge and a non-tree edge in a
$(\lambda,w)$ plane is represented as a common tangent to the convex hulls of
the tree and non tree edges in the dual plane.

\begin{figure}[t]
\begin{center}
\includegraphics[angle=270, width=82mm]{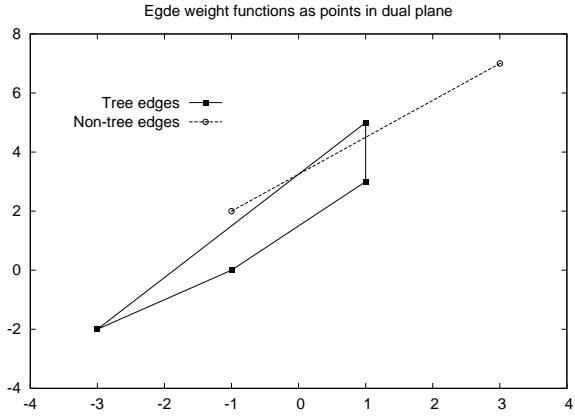}
\caption{Intersecting convex hulls of tree and non-tree edges.}
\label{rconvex-a}
\end{center}
\end{figure}

This method is very efficient when there is a clear separation among the convex
hulls of tree and non-tree edges. Otherwise, the convex hulls may overlap. For
example, Figure~\ref{rconvex-a} shows the corresponding points in the dual plane
for the edge weight functions between time $t=1$ and $t=2$ for the ST network
shown in Figure \ref{mdt-a}. Figure~\ref{rconvex-a} shows the overlapping of the
convex hulls formed by the tree and non-tree edges. However, this kind of
separation can be seen inside a single cycle and thus would require to creating
and maintaining $O(m-n+1)$ convex hulls (one for each fundamental cycle), where
$m$ is the number of edges and $n$ is the number of node.

Moreover, work done in field of non-stationary candidate ranking assume
separable edge weights, i.e, they assume that there is no correlation between
the different weights of an edge at different time instants. Due to this
assumption the kinetic algorithms do not address the situation when change
points and intersection points overlap. Change points are those time instants
where an edge weight function changes its slope and intersection points are
those time instants where two or more edge functions intersect, i.e., they have
same edge weight. For example, in Figure~\ref{rconvex-b} edge (3,5) has  change
points at $t=3$ and $t=2$ and edge weight functions of (1,3) and (2,4) intersect
at $t=1.5$, whereas the weight functions of edges (1,3) and (2,3) both change
and intersect at the same point. At such points the convex hulls may not have a
unique common tangent. This implies that the MST cannot be changed even if it is
required.\\

\section{TSMST Computation Algorithms}
\label{tca}
In this chapter, we present two algorithms for computing the TSMST of a spatio-temporal
network. Consider again the sample network shown in Figure~\ref{mdt-a}
and its edge-weight function plot in Figure~\ref{rconvex-b}. The following
observations can be inferred from edge weight function plot.

\begin{obs}
	\label{ob1}
Consider any two consecutive (with respect to time coordinate) intersection
points of the edge weight functions. These time coordinates form an
edge-order-interval. Within this time interval, all the edge weight functions
have a well defined order.
\end{obs}

\begin{obs}
	\label{ob2}
Using the ordering of edge weights within an edge-order-interval, an MST for
this interval can be built using a standard greedy algorithm such as Kruskal's~\cite{CPROP}
or Prim's~\cite{CPROP}.
\end{obs}

\begin{obs}
	\label{ob3}
There will be a single MST for the entire edge-order-interval.
\end{obs}

\subsection{Time Sub-Interval Order (TSO) Algorithm}

The time sub-interval order (TSO) algorithm is designed (Algorithm 1) using
Observation \ref{ob1}, Observation \ref{ob2} and Observation \ref{ob3}. The TSO
algorithm starts by determining all edge-order-intervals. First, it
computes all the intersection points of the edge weight functions. The
intersection points are then sorted with respect to time coordinate. This sorted
list of all intersection points is termed as \emph{edge-order series}. Then the set
$\delta$ of all edge-order-intervals is built from the edge-order series. Each
item in $\delta$ is a pair obtained by picking consecutive elements from the
edge-order series. 

The algorithm then computes MST for the first
edge-order-interval of $\delta$. Next, MST for the next interval in $\delta$ is
computed and compared with the previous MST. If the current MST is same as the
previous MST, then this interval is combined with the time-sub-interval of the
previous MST. Otherwise, the previous MST is output along with its
time-sub-interval. The previous MST forms the TSMT for that particular
time-sub-interval (which was output). This process continues until there are no
more intervals left in $\delta$. Since the algorithm outputs the previous TMST
every time, the last TMST has to be output separately.

\begin{algorithm}[t]
\caption{Time Sub-interval Order Algorithm for finding TSMST}
\label{alg1} 
\begin{algorithmic}[1]
\STATE Determine the intersection points of edge weight functions for all $e \in E$
\STATE Sort the intersection points with respect to time coordinate to form an edge-order series 
\STATE Pick the consecutive elements of the edge-order series to build the set $\delta$
\STATE Compute MST for the fist edge-order-interval of $\delta$ and denote as previous MST
\STATE Initialize the time-sub-interval of the previous MST to the first edge-order-interval  
\FORALL{remaining edge-order-intervals $\omega$ in $\delta$}
\STATE Pick the next interval from $\delta$
\STATE Build Minimum Spanning Tree using any greedy algorithm like Prim's or Kruskal's
\IF{current MST differs from previous MST}  
\STATE Output the previous MST (which is now TMST) along with the time-sub-interval and make current MST as the previous MST
\STATE Set the time-sub-interval of previous MST to $\omega$ 
\ELSE 
\STATE Combine $\omega$ with the time-sub-interval of previous MST
\ENDIF
\ENDFOR
\STATE Output the previous MST (TMST) along with its time-sub-interval
\end{algorithmic}
\end{algorithm}

\begin{figure}[t]
\begin{center}
\includegraphics[width=88mm]{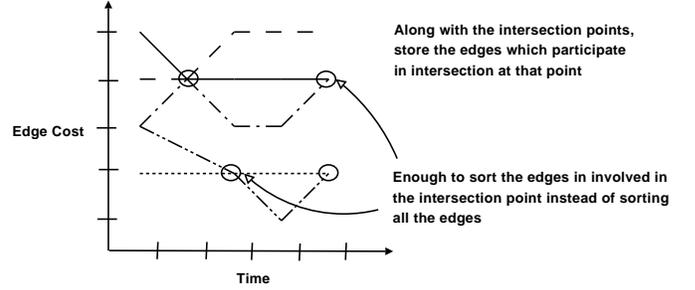} 
\caption{Sorting only the edges involved in intersection.}
\label{lsort}
\end{center}
\end{figure}

\begin{obs}
	\label{ob4} 
The time-sub-interval order (TSO) algorithm  can be improved by sorting only the edges involved in an intersection instead of sorting all the edges in each edge-order-interval (see Figure \ref{lsort}). This is true because only the edges involved in intersections can change their relative ordering.
\end{obs}

The TSO algorithm computes a MST for each intersection point. This would incur
an unnecessary overhead if the MST does not change at a intersection point
(e.g., see Figure \ref{nworse}).

Moreover, if the edges involved in intersection are from different bi-connected
components, the MST will not change (Proposition~\ref{cor1}). A bi-connected component of a
graph is a maximal set of edges such that any two edges in the set lie on a common cycle~\cite{SMST1}. For example, 
edges $(3,5)$ and $(1,2)$ of the network shown in Figure~\ref{mdt-a} belong to different bi-connected components, therefore intersection of their weight functions 
can never produce a change in MST. On the other hand, intersection of weight functions of edges in the same bi-connected component
(e.g., edges $(1,2)$, $(2,4)$, $(3,4)$, $(1,3)$, $(2,3)$ in Figure \ref{mdt-a}) may change the MST.

Similarly, intersection of edge weight functions of two or more tree edges or non-tree edges do not change the MST (Proposition \ref{lem1} and Proposition \ref{lem2}).
For example, in Figure~\ref{rconvex-b} intersection of weight functions of edges $(1,3)$ and $(2,3)$ at time $t=2$ does not change the MST because both are non-tree edges. 
Likewise, intersection of weight functions of edges $(1,2)$ and $(2,4)$ at time $t=2$ (see Figure \ref{rconvex-b}) does not change the MST because both are tree-edges.

Furthermore, if only one tree edge and one non tree edge (belonging to only one common cycle) are
involved in intersection, we can exchange those edges in tree. For example see
Figure \ref{c1}. Here, the weight function of edges $(1,2)$ (a non-tree edge)
and $(5,6)$ (a tree edge) intersect at time $t=1.5$. Thus instead of building a
new MST at $t=1.5$ (as done by TSO), exchanging the edges involved in
intersection  in current MST would give the new MST (see Figure \ref{c1}). These
ideas are presented formally in the following propositions. These are used in
designing an incremental algorithm for computing the TSMST. 

\begin{figure}[t]
\begin{center}
\includegraphics[width=86mm]{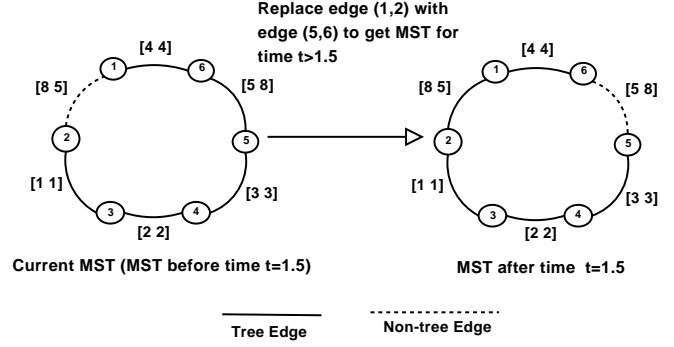}
\caption{Edge exchange inside a cycle.}
\label{c1}
\end{center}
\end{figure}

\begin{proposition}
	\label{lem1}
	The intersection of the edge weight function of two non-tree edges at any time instant will not affect the minimum 
spanning tree, i.e., the tree will not change.\label{lem1}
\end{proposition}
\begin{proof} 
Consider the intersection of the edge weight function of two non-tree edges $e_x$ and $e_y$ at a time instant
$t=\alpha$. Let $O_1$ = $e_1 e_2 \ldots e_i e_x e_y e_{i+3} \ldots e_m$ be the ascending order of weights of the edges 
before time $t=\alpha$ and $O_2$ = $e_1 e_2 \ldots e_i e_y e_x e_{i+3} \ldots e_m$ after time $t=\alpha$. Here, $m$ is number of edges in the graph.
Since all the edges have distinct weights, therefore, MST of the graph induced by the edges $e_j$, where $j \leq i$, will be same 
for both the series $O_1$ and $O_2$ (using Lemma \ref{lem4}). Now due to the intersection, edges $e_x$ and $e_y$ have changed their relative ordering.
Previously, i.e., before time $t=\alpha$, edge $e_x$ was not present in the MST. But now its weight has further increased so it will again not be in the MST. Now, the weight 
of the edge  $e_y$ has decreased but it is still more than the next lighter edge (since weight functions of only $e_x$ and $e_y$ have intersected). Therefore, $e_y$ will
again not be in the MST. Thus the intersection of the edge weight functions of two non-tree edges does not affect the MST.      
\end{proof}

\begin{proposition}
	\label{lem2}
	The intersection of the edge weight function of two tree edges at any time instant will not affect the minimum 
spanning tree, i.e., the MST will be the same. \label{lem2}
\end{proposition}
\begin{proof}
Consider the intersection of the edge weight function of two tree edges $e_x$ and $e_y$ at a time instant
$t=\alpha$. Let $O_1$ = $e_1 e_2 \ldots e_i e_x e_y e_{i+3} \ldots e_m$ be the ascending order of weights of the edges 
before time $t=\alpha$ and $O_2$ = $e_1 e_2 \ldots e_i e_y e_x e_{i+3} \ldots e_m$ after time $t=\alpha$. Here, $m$ is number of edges in the graph.
Since all the edges have distinct weights, therefore, MST of the graph induced by the edges $e_j$, where $j \leq i$, will be same 
for both the series $O_1$ and $O_2$ (using Lemma \ref{lem4}). Now due to the intersection edges $e_x$ and $e_y$ have changed their relative ordering.
Previously i.e., before time $t=\alpha$, edge $e_y$ was present in the MST. But now its weight has decreased so it will again be in the MST. Now, the weight 
of the edge  $e_x$ has increased but it is still less than the next heavier edge (since weight functions of only $e_x$ and $e_y$ have intersected). Therefore, $e_x$ will
again be in the MST. Thus the intersection of the edge weight functions of two tree edges does not affect the MST. 
\end{proof}
   
\begin{proposition}
	\label{cor1}
	The intersection of the edge weight functions of two edges belonging to different bi-connected components can never
 change the MST.\label{cor1}
\end{proposition}
\begin{proof}
Consider a graph $G$ with two bi-connected components $G_1$ and $G_2$ such that $G$ = $G_1 \cup G_2$. The bi-connected components $G_1$ and $G_2$ can at most share one 
vertex (by the maximality of bi-connected components \cite{SMST1}). Consider a spanning tree of $G$, $T_G$ = $T_{G_1} \cup T_{G_2}$. Now in order to determine the minimum spanning tree of 
$G$ we have to minimize the total cost of $T_G$. The minimum value is attained when both $T_{G_1}$ and $T_{G_2}$ have their minimum values. But the minimum value of the total cost
of $T_{G_1}$ ($T_{G_2}$) is attained when $T_{G_1}$ ($T_{G_2}$) is the MST of $G_1$ ($G_2$). This implies that MST of $G$ is the union of the MSTs of its 
individual bi-connected components. 

Consider the intersection of the weight function of an edge $e_x$ belonging to $G_1$ with an edge $e_y'$ belonging to $G_2$. Let $O_1$ = $e_1 e_2 \ldots e_i 
e_x \ldots e_{m_1}$ be the ascending order of the edge weights of edges in $G_1$ and $O_2$ = $e_1' e_2' \ldots e_i' e_j' \ldots e_{m_2}'$ be the ascending order of the edge weights of
edges in $G_2$. Here $m_1$ is the number of edges in the bi-connected component $G_1$ and $m_2$ is the number of edges in the bi-connected component $G_2$. Now due to this intersection 
there is no change in $O_1$ or $O_2$. Therefore the MST of $G_1$ (or $G_2$) will not change. This implies that the MST of $G$ will not change.  
\end{proof} 

\begin{proposition}
Intersection of weight functions of two or more tree-edges (or non-tree edges) from one bi-connected component with two or more non-tree edges (or tree edges) from 
another bi-connected component cannot change the MST. \label{cor2}  
\end{proposition}

\begin{proof}
Consider the intersection of weight functions of two tree edges $e_1$ and $e_2$ of bi-connected component $B$ with two non-tree edges $e_1'$ and $e_2'$ of bi-connected component $B'$.
Using Proposition \ref{cor1}, we conclude that intersection of weight functions of edges belonging to different bi-connected components do not change the MST. Now, inside each of the 
bi-connected components $B$ and $B'$ either only tree edges are involved or only non-tree edges are involved. Using Proposition \ref{lem1} and Proposition \ref{lem2} we conclude that  
these intersections cannot change the MSTs within the bi-connected components. Now, since MSTs of individual bi-components did not change therefore the MST of the entire 
graph will also not change. This is because the MST of graph is the union of the MSTs of its individual bi-connected components (Proposition \ref{cor1}).
\end{proof}

\begin{figure}[t]
\begin{center}
\includegraphics[width=79mm]{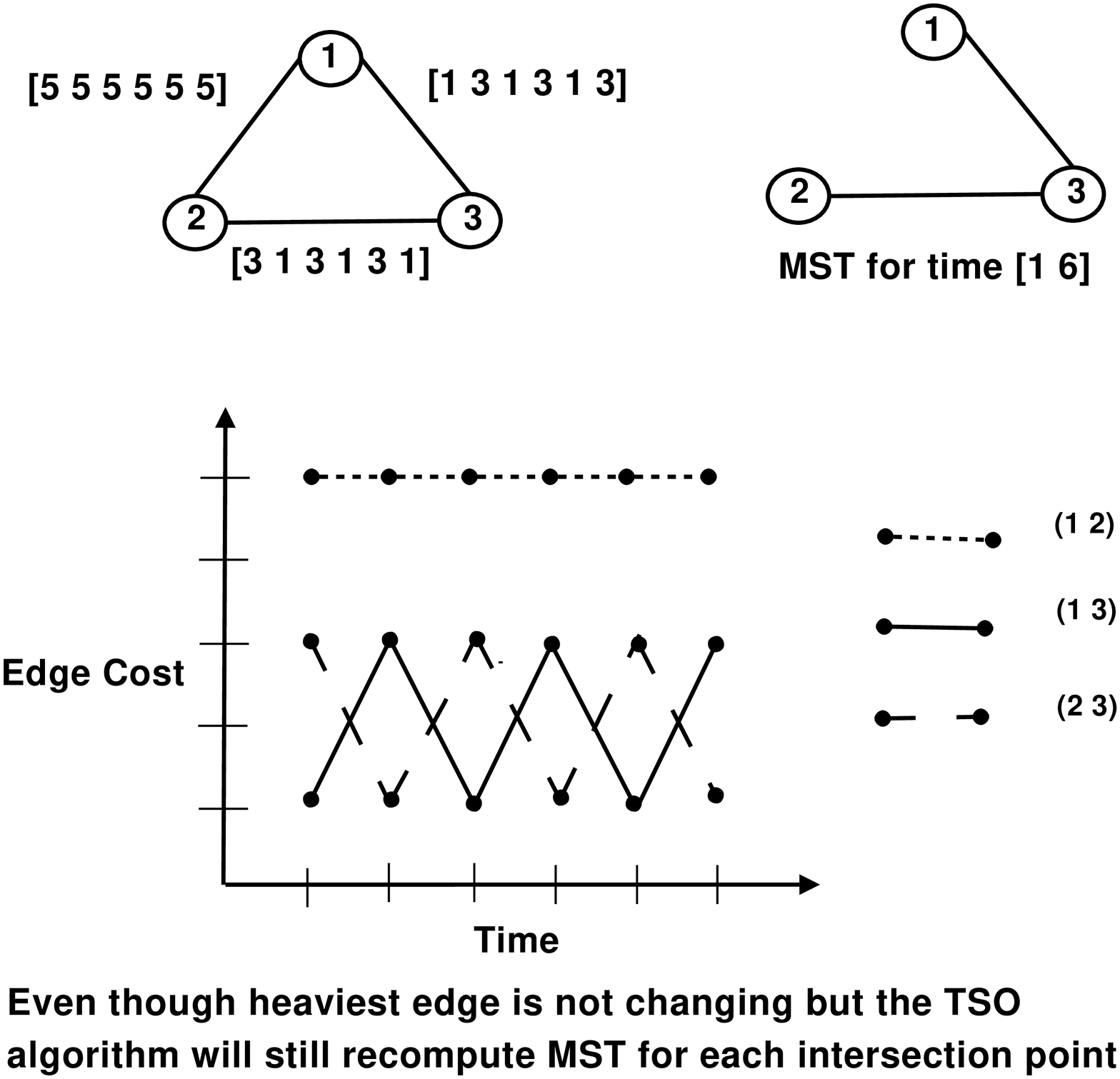} 
\caption{Worst case behavior of TSO algorithm.}
\label{nworse}
\end{center}
\end{figure}

\subsection{Edge Intersection Order (EIO) Algorithm}

Here we present an incremental algorithm for computing TSMST. The edge
intersection order (EIO) algorithm starts by computing the MST of the network at
time $t=1$ and then continues to update the tree, only if necessary, at each
intersection point. The intersection points are processed in increasing order of
their time coordinates. Through preprocessing, some additional information about
the edges is stored while computing the MST at time $t=1$. This information is
used to prune the intersection points which are guaranteed not to cause any
change in the MST.

The modified reverse-delete (Algorithm~\ref{alg2}) is used to compute the MST at
time $t=1$. Figure \ref{cc} shows the execution trace of the algorithm. The
algorithm first computes the depth first search (DFS) tree \cite{SMST1} of the given ST
network. A non-tree edge $ne=(f_s,f_e)$, where $f_e$ is the ancestor of $f_s$,
is chosen. Now, edge $ne$ and edges seen while following the parent pointers
from node $f_s$ to $f_e$ and $ne$ form a cycle. This cycle is termed as \emph{fcycle}.
The heaviest edge of this fcycle is deleted. For example, in Figure \ref{cc},
edge (2,4) is a non-tree edge where node 2 is an ancestor of node 4 (considering
the DFS tree to be rooted at node 2). Now on following the parent pointers from
node 4 to node 2, edges (4,3), (3,2) and (2,4) form a fcycle with edge (2,4)
being the heaviest. This edge is deleted and an entry in a table called fcycle
table is made for FC1 with edges (4,3), (3,2) and (2,4). A DFS tree of the
remaining edges is computed. A non-tree edge is picked up and the heaviest edge
in its fcycle is deleted and a corresponding entry in the fcycle table is made.
This process continues until only $n-1$ edges are left in the network. At this
point there will no non-tree edges after constructing the DFS tree, i.e., all
edges are tree edges. These edges form the MST of the network. They are marked
as tree edges and all other remaining edges of the original network are marked
as non-tree edges.\\ 


\noindent\textbf{Data structures used by EIO algorithm}
The following data structures are used by the EIO
algorithm and the modified reverse-delete algorithm. This information is
pre-computed before the EIO algorithm starts processing each intersection point:
\begin{itemize}
	\item \textbf{Edge-Table:} A look up table storing the details of each edge.
		This table contains, for each edge, a unique edge-id, the two nodes it
		connects, the bi-connected component it belongs to, and a bit vector
		storing the fcycles of which the edge is a member.   
	\item \textbf{Fcycle-Table:} A look up table storing the details of each
		fcycle observed while constructing the minimum spanning tree at time
		$t=1$. Each entry of this table contains a unique fcycle-id and a list
		of its member edges (edges column).
\end{itemize}

The MST of the network is stored as a bit vector of length equal to the number
of edges in the network. The bit vector would contain a bit corresponding to
each edge of the network and the edges belonging to MST would be set to one. The
edge-table and fcycle-table are indexed using hashing. The information stored in
fcycle-table can also be computed at runtime by adding the non-tree edge to the MST
and determining the fcycle. This would incur an additional $O(n)$ cost each
time. This is avoided by storing the information while constructing MST for time
$t=1$. The fcycles column of the edge-table is filled by traversing through the
fcycle-table. Each fcycle-id from the fcycle-table is chosen and the bit
corresponding to this fcycle-id is set for all the edges in the edges column of
this entry of fcycle-table. The bi-connected component column of the edge-Table
is filled  in linear time using the algorithm given in \cite{AHO}.

\begin{algorithm}[t]
\caption{Modified Reverse-Delete Algorithm for MST}
\label{alg2}
\begin{algorithmic}[1]
\STATE Find depth first search tree of the given network
\WHILE{non-tree edge is present} 
\STATE Pick a non-tree edge and determine all the edges in the fcycle
\STATE Delete the edge which has maximum weight. Record the fcycle in $fcycle-table$
\STATE Find the DFS tree of the remaining network
\ENDWHILE
\STATE The remaining edges form the minimum spanning tree
\end{algorithmic}
\end{algorithm}

While considering an intersection point, two levels of filters are applied to
prune the intersection points that cannot cause any change in MST. If all the
edges involved in an intersection are either non-tree edges or tree edges, then
it can be pruned using Proposition \ref{lem1} and Proposition \ref{lem2} respectively.
Similarly, if all the edges involved in an intersection belong to different
bi-connected components, then it can be pruned using Proposition \ref{cor1}.

After applying these filters, the edges are grouped by their bi-connected
component number. Now, within each group we again check if the edges are only
tree edges or non-tree edges. This is important in cases when the weight
functions of two or more tree (or non-tree) edges of one component and one or
more non-tree (or tree) edges of another component intersect at a point. These
kinds of intersection points cannot cause any change in a tree (Proposition \ref{cor2}). 
Hence, the filters are applied again after the edges are grouped by their bi-component
number. After applying these filters if the intersection point is not pruned, 
we check if the relative order of edges weights before and after the
intersection point are same. The intersection point can be pruned safely if
the relative order is same. This can happen when weight functions of two or more
edges touch each other at change points (see Figure~\ref{rconvex-b}) without
changing their relative order.

If an intersection point is not pruned after applying all the filters, then
a new MST is made by making changes to the previous MST. If only two edges (per
bi-connected component) are involved in the intersection and they are part of
only one common cycle, i.e., they are not part of any cycle except the one which
is common, then we can directly exchange the edges in the tree, i.e., make the
heavier edge between them as non-tree edge and the lighter edge as tree edge. In all
other cases we add each of the non-tree edges involved in the intersection to
the MST and delete the heaviest edge from their respective fundamental cycles.
We can check whether two edges are part of only the common cycle by performing
an AND operation on the fcycles column of the two edges and check the number of
bits set to 1 in the result. Note that in the two-edges intersection case
discussed above, adding the non-tree edge to the tree and deleting the heaviest
edge from its fundamental cycle would still give the correct MST. The
information gathered during initial construction was used to save this
unnecessary re computation. The start time of the time-sub-interval of the new
MST and end time of the time-sub-interval of previous MST are set to the time
coordinate of intersection point at which the MST changed. The previous MST is
output along with its time-sub-interval. The previous MST forms the TMST for
that particular time-sub-interval.

We now describe the terminology used in the EIO algorithm.   	     
While considering an intersection point at time $t$, $S_t$ is the set of all
intersection points whose time coordinate is $t$. An intersection point in $S_t$
is denoted as  $p_r$, where $r \in [1,2, \ldots |S_t|]$. The set of all $S_t$'s
form $\Gamma$. $B_{p_r} = [b_1, b_2, \ldots]$ is the set of groups obtained
after grouping the edges involved in the intersection by their bi-connected
components. $s_i$ and $s_j$ are the decreasing order of edges weights (of a
particular bi-connected component group) before and after the time of the
intersection point. The edge intersection order algorithm is formally presented
as Algorithm~\ref{alg3}.\\

\begin{figure}[t]
\begin{center}
\includegraphics[width=94mm]{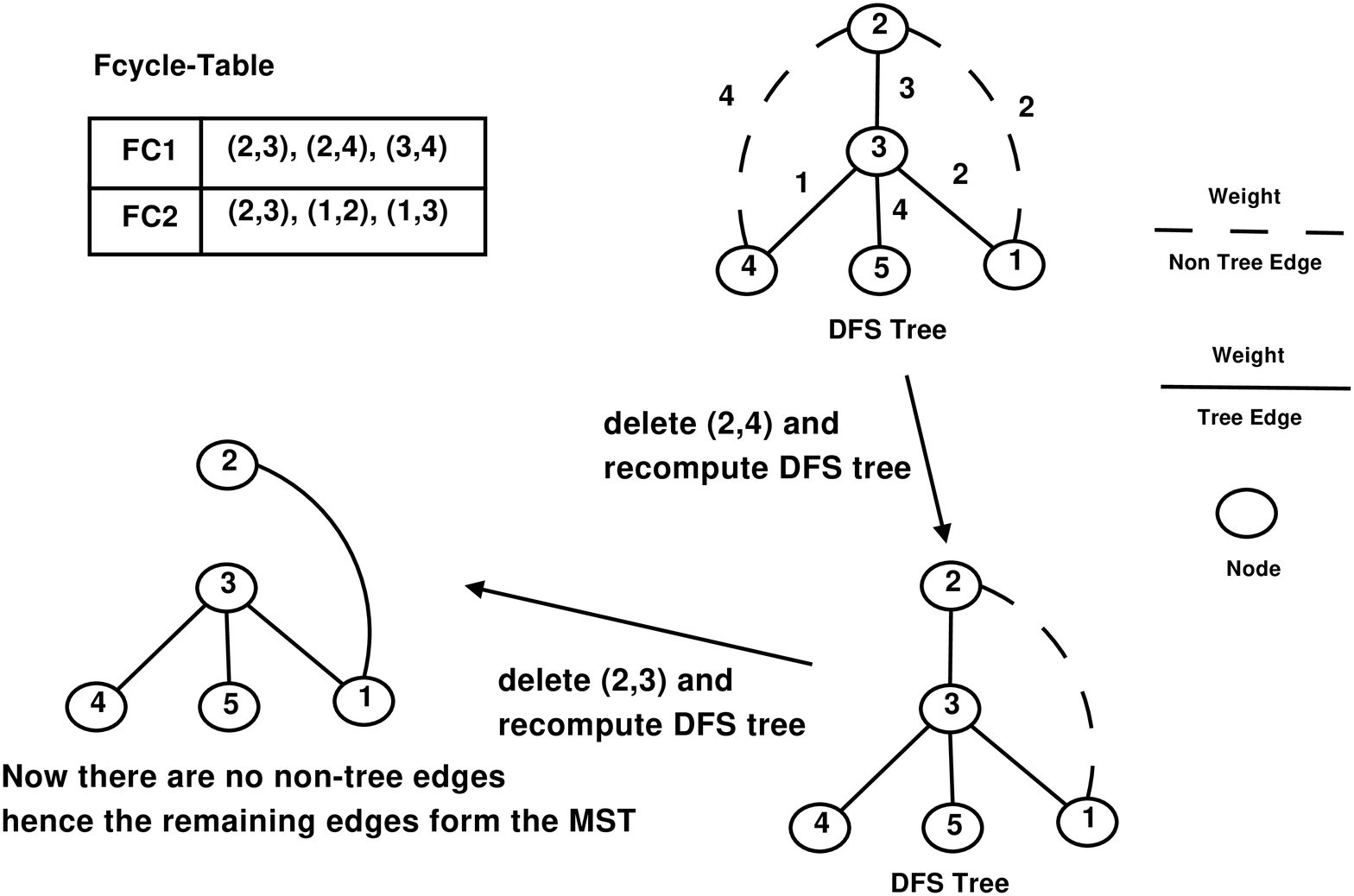}
\caption{Trace of modified reverse delete algorithm for building MST.}
\label{cc}
\end{center}
\end{figure}

\noindent\textbf{Execution Trace of EIO Algorithm} We next present an execution trace of
the EIO algorithm for the example in Figure~\ref{mdt-a}. The algorithm starts by
computing the bi-connected components of the given graph and filling the
corresponding columns in the edge table. This can be done in linear time
\cite{AHO}. The algorithm then determines the intersection points of the edge
weight functions. They are then sorted on the time coordinate to form the
edge-order series. Minimum spanning tree at $t=1$ is built using Algorithm
\ref{alg2} and the fcycle-table is populated. Step by step execution of
Algorithm \ref{alg2} and the filling of the fcycle-table is shown in
Figure~\ref{cc}.

After building the MST for time $t=1$ a bit vector is created by setting the
bits corresponding to tree-edges to 1 (MST(I) in Figure \ref{re}). The
edge-table and MST(I) is shown in Figure \ref{re}. After that the remaining
intersection points belonging to the remaining intervals of $\delta$ are grouped
by their time coordinate to build $S_t$ (here, $S_t$ is the set of all
intersection points with time coordinate as $t$). The set of all $S_t$'s form
$\Gamma$. The algorithm then updates the MST, if necessary, by exchanging the
edges at the intersection points. After $t=1$ the next interval in $\delta$
starts $t=1.25$. There are two pairs of edges intersecting at $t=1.25$. Both
these pairs involve only tree edges or non-tree edges, and thus, these are
pruned by the if conditions at line 10 and line 13.\footnote{Either the previous
MST or the current MST may be used here because any edge can be involved in only
one intersection at a time.} Similarly, the intersection of edge (3,5) and (2,3)
is pruned since they belong to different bi-connected components. 
\begin{figure}[t]
\begin{center}
\includegraphics[width=70mm]{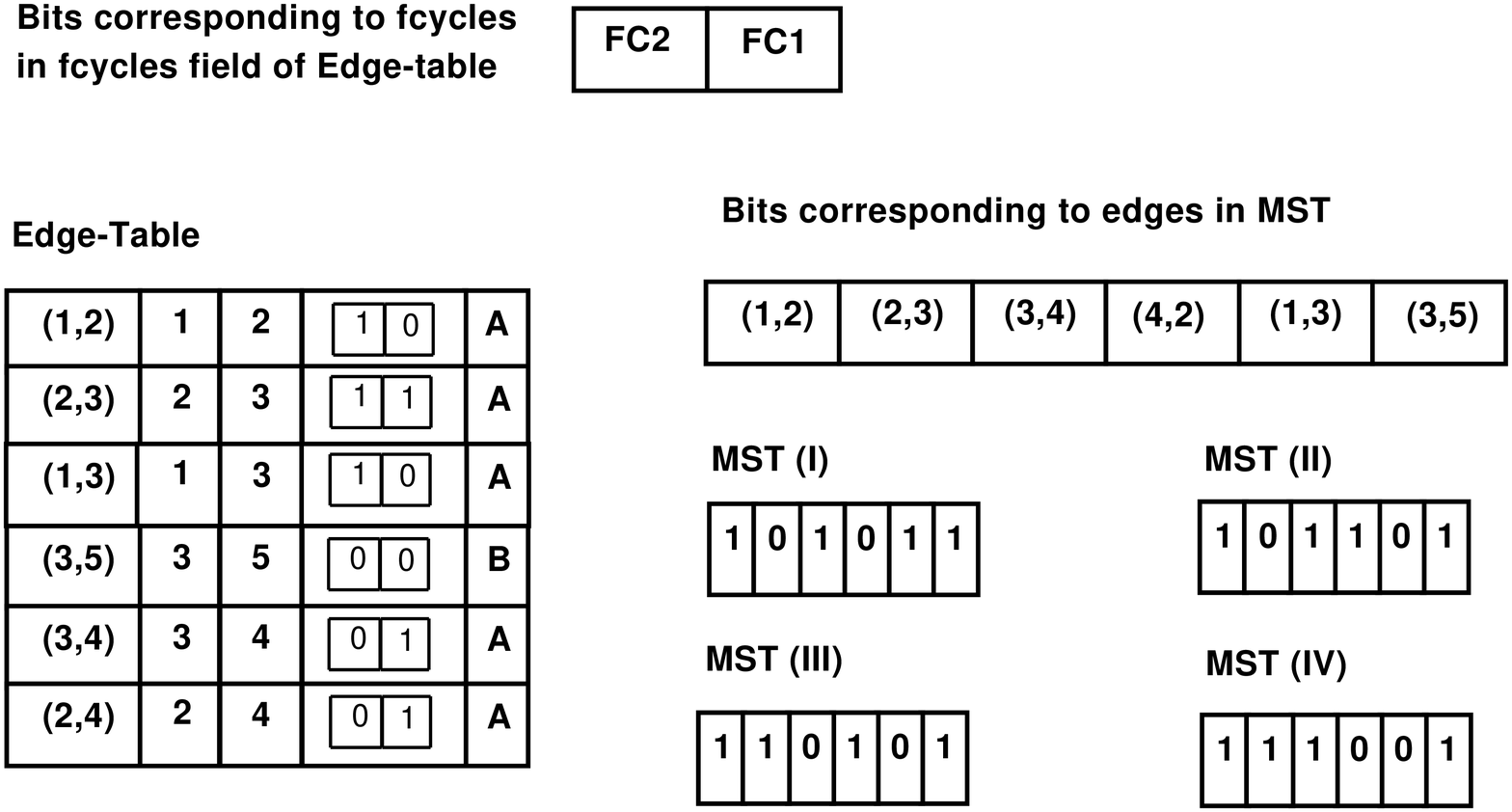}
\caption{Trace of EIO Algorithm for ST network shown in Figure \ref{mdt-a}.}
\label{re}
\end{center}
\vspace{-4mm}
\end{figure}

Now for the intersection between the edges (2,4) and (1,3) at t=1.5, the condition in line
30 is evaluated by performing an bitwise AND operation of their respective
fcycle-id columns. Since the edges (2,4) and (1,3) do not share a cycle (as per
fcycles column of edge-table) the non-tree edge (2,4) is added to  current MST
and heaviest edge in the fundamental cycle is deleted and MST is updated (now it
will be MST(II) in Figure \ref{re}). This step is essential because even if the
edges do not share a cycle as per fcycles column of  edge-table they still lie
on a common cycle \cite{SMST1}. The next intersection point
involving tree and non-tree edge of same component is at t=2.66 between (3,4)
and (2,3). Here, the edges (3,4) and (2,3) share a cycle (if condition at line 30
is true) and thus are exchanged to create MST(III). Similarly the next exchange
is at t=3.66. Since the EIO algorithm outputs the older MST after every change,
the last TMST has to be output separately after the main loop.

\begin{algorithm*}[ht]
\caption{Edge Intersection Order (EIO) algorithm for finding TSMST}
\label{alg3}
\begin{algorithmic}[1]
\STATE Compute the bi-connected components of the given graph and update bi-connected column of edge-table
\STATE Determine the intersection points of the edge weight functions 
\STATE Sort the intersection points with respect to time coordinate 
\STATE Find the MST at $t=1$ (let this be previous MST) using Algorithm \ref{alg2} and populate the fcycle-table
\STATE Set the start time of time sub interval of previous MST to $1$ 
\STATE Construct $S_t$s from the intersection points
\FORALL{$S_t \in \Gamma$}
\STATE Set MST-change flag and new-MST-exist flag to FALSE
\FORALL{intersection points $p_r \in S_t$}
\IF{all edges are tree edges or non-tree edges}
\STATE continue
\ENDIF
\IF{all edges belong to different bi-connected components}
\STATE continue
\ENDIF
\STATE group the edges into bi-connected components (build set $B_{p_r}$)
\FORALL{$b_i \in B_{p_r}$} 
\IF{all edges are tree edges or non-tree edges}
\STATE continue
\ENDIF 
\STATE Determine $s_i$ and $s_j$.
\IF{$s_i$ and $s_j$ are same}
\STATE continue
\ENDIF
\STATE Set MST-change flag to TRUE (all subsequent steps will change MST)
\IF{new-MST-exist is FALSE}
\STATE Create a new MST (i.e., create a new bit vector and assign the same values as of previous MST), let this be current MST 
\STATE Set new-MST-exist flag to TRUE
\ENDIF
\IF{only two edges intersect and they are part of only the common cycle}
\STATE update the current MST (set the bit corresponding to heavier edge to 0 and lighter edge to 1).
\ELSE
\FORALL{non tree edges in $s_j$}
\STATE find their fundamental cycle and delete heaviest edge
\STATE update the current MST
\ENDFOR
\ENDIF
\ENDFOR
\ENDFOR
\IF{MST-change flag is set to TRUE}
\STATE Output the previous MST (TMST) with its time-sub-interval and make current MST as previous MST 
\ENDIF
\ENDFOR
\STATE Output the previous MST (TMST) along with its time-sub-interval
\end{algorithmic}
\end{algorithm*}

\subsection{Relaxation of Edge Presence Assumption}
In this section, we relax the assumption that an edge is always present. The absence of an edge can be represented by 
assigning time intervals during which the edge is absent. Such a case arises when some of the sensor nodes move very far from other sensor nodes.
Consider a case when an edge $e$ is not present during the time interval [2.5 2.9]. Further assume that its weight at time instants $t=1,2,3,$ 
and $4$ is $2, 3, 6,$ and $8$ respectively.   
The edge absence information can be combined with the edge weight time series by expanding the series to define weights at smaller time intervals. For example, 
in the above case, the edge weights would be defined at time instants $t=1.0,1.1,1.2 \ldots 2.4, 2.5-\epsilon$ and then at time
instants $t= 2.9 + \epsilon, 3.0, \ldots 3.9, 4.0$. The edge weight at these new intermediate time instants can be determined by interpolating 
original weight function \footnote{Since the edge function is linear,
interpolation works nicely.}. 

The absence of an edge for a certain time period affects the MST and the bi-connected components of the graph.
If a non-tree edge does not exist after a certain time instant then there will no change in the MST. On the contrary, if a tree edge 
ceases to exist then it is replaced by a non-tree edge of lowest weight such that the MST remains connected and optimal. This is done by
adding each of the non-tree edges (in increasing order of their weights) to the MST and if the tree edge (which is absent) in found in the cycle created, then it is 
replaced by that particular non-tree edge. If multiple tree edges are absent during a particular time interval then each of these tree edges 
is replaced by a non-tree edge. In such cases, each of the non-tree edges (in increasing order of their weights) is added to the MST and if any of the tree edges 
is found in the cycle created by the addition of non-tree edge, then it is replaced by that particular non-tree edge. The tree edges are stored in a hash table for this case. 
Now, when an edge reappears (either tree or non-tree), it is added to the MST and the heaviest edge is deleted from the cycle created. For practical purposes, we assume 
that the graph is never disconnected due to disappearance of edges. Therefore, at most $m-n+1$ edge can be absent during any particular time interval. 

When an edge is absent during a certain time interval, the bi-connected component information of the edges change. This information can be 
recomputed in linear time using the algorithm given in \cite{AHO}. This updated information may be useful for pruning some of the 
intersection points which may not be pruned otherwise. The bi-connected component information has to be recomputed again after the edge appears.
Re-computation of bi-connected components can be delayed until the next intersection point is encountered. This is useful in cases when there are no 
intersections during the period of absence. The bi-connected component information is only used by the EIO algorithm.

\section{Analytical Evaluation}
\label{eval}
The correctness proof and the asymptotic analysis of the edge-intersection order algorithm and the time sub-interval order algorithm are presented in this
chapter. 
\subsection{Analysis of TSO algorithm}
\begin{theorem}
The time-sub-interval order (TSO) algorithm is correct.
\end{theorem}
\begin{proof}
The correctness of the TSO algorithm can be established by using Observation \ref{ob1}. The minimum spanning tree changes
 whenever there is a change in ordering of the edges. The TSO algorithm addresses this by recomputing the minimum spanning
 tree for each interval in $\delta$. This means that it recomputes the ordering of edges again after every intersection 
point. Moreover, from Observation \ref{ob1}, we can say that the ordering or edges do not change inside an edge order 
interval. This proves that the TSO algorithm is correct.
\end{proof}

\noindent\textbf{Asymptotic analysis of TSO algorithm} Since the edge weight is assumed to vary
linearly between any two time instants in a time series, there can be at most
$O(m^2)$ (where $m$ is the number of edges) intersections among the edge weight
functions in the worst case. Now if this happens between all time instances in
the entire time horizon $[1, \ldots K]$, the total number of intersections
will be $O(m^2K)$. The time needed to sort all the intersection points is
$O(m^2K\log(m^2K))$. For each intersection point TSO recomputes the MST
which takes $O(m\log m)$ per intersection point. Thus, the total time complexity
of the TSO algorithm is $O(m^3K\log m + m^2K\log(m^2K))$.\\ 

\noindent\textbf{Relaxation of edge presence assumption}
Now, consider the case when the edge presence assumption is relaxed. Even though the
number of time instants in the edge weight time series increases,
the total number of intersection points do not change because we are not adding any new
weight functions. Let $L$ be the number of disjoint time intervals during which any of the edges is absent.
During each of these $L$ intervals, if any of the tree edges is absent then it has to replaced by a non-tree edge.
The replacement involves observing the cycles created by each of the $O(m-n+1)$ 
non-tree edges in increasing order of their weights. Since there can be at most $n-1$ tree edges, this takes $O(mn + m\log m)$ time.
Here, each tree edge is not processed individually. All the tree edges are stored in a hash table and are checked in the cycles created by the addition of each non-tree edge.
Then after an edge (either tree or non-tree) reappears, it is again added to MST and the heaviest edge is deleted from the cycle created. Now, there can be at most $m-n+1$ edges absent 
during a time interval (assuming that the graph always remains connected). Therefore, this step takes $O(mn)$ time.
Therefore the total complexity is $O(Lmn + Lm\log m)$. Thus, the overall complexity of 
the TSO algorithm becomes $O(m^3K\log m + m^2K\log(m^2K) + Lmn + Lm\log m)$.


\subsection{Analysis of Modified Reverse Delete Algorithm}

\begin{theorem}
The modified reverse delete algorithm produces a MST.
\end{theorem}
\begin{proof}
(By contradiction) Let $T_1$ be the spanning tree generated by the algorithm and let $T_2$ be the MST. Let $OT_1$ = $e_1 e_2 \ldots e_i e_{i+1} \ldots e_{m-n+1}$ be the 
increasing order of edge weights of $T_1$. Let $OT_2$ = $e_1' e_2' \ldots e_i' e_{i+1}' \ldots e_{m-n+1}'$ be the increasing order of edge weights 
of $T_2$. Without loss of generality, assume that $e_j$ is same as $e_j'$, $\forall j \leq i$. Further assume that $e_{i+1} > e_{i+1}'$. Now consider 
the cycle generated by adding the edge $e_{i+1}'$ to $T_1$. Only one such cycle can be created as $T_1$ is a spanning tree of the graph. Now if this cycle contains only the edges $e_j$ where $j \leq i$, then it 
implies that there is a cycle in $T_2$ as $e_j$ is same as $e_j'$, $\forall j \leq i$ (a contradiction).  Now, if the cycle contains some edges $e_j$ where $j \leq i$ and $e_{i+1}$, then 
the algorithm would have chosen edge $e_{i+1}$ to delete instead of $e_{i+1}'$ because the algorithm deletes the heaviest edge of the cycle. Consider the case when the cycle contains the 
edges $e_j$ where $j \geq i+1$. Let $e_j$, where $j \neq i+1$, be any edge of the cycle. Then, $e_{i+1}' < e_{i+1} < e_j$. Again, the algorithm would have chosen edge $e_{i+1}$ instead of $e_{i+1}'$. 
Therefore, we can conclude that $T_1$ and $T_2$ are same.
\end{proof}

\noindent\textbf{Asymptotic analysis of Modified Reverse Delete Algorithm}
The algorithm determines the DFS tree of the graph in each iteration. This takes $O(m+n)$ where $m$ is the number of edges and $n$ is the number of nodes. Finding the heaviest edge in a 
fcycle takes $O(n)$ time. This happens when the fcycle contains all the nodes of the graph. After each iteration the number of edges decreases by one. Therefore, the total number of iterations 
required are $m-n+1$ (one for each non-tree edge deleted). Therefore, the total time taken is given by the sum of the series $(m+n) + (m-1+n) + \ldots + (2n-1)$. This series has $m-n+1$ terms. 
Thus, the overall complexity of the algorithm is $O(m^2)$.

\subsection{Analysis of EIO algorithm}

\begin{theorem}
The edge intersection order (EIO) algorithm is correct.
\end{theorem}
\begin{proof}
The EIO algorithm prunes an intersection point if only tree edges or non-tree edges are involved. 
Proposition \ref{lem1} and Proposition \ref{lem2} show the correctness of this filtering step.
Similarly, an intersection point is pruned if all the edges involved in it belong to different components. The correctness
of this filtering step is evident from Proposition \ref{cor1}. An intersection point is also pruned if two or more tree
edges (or non-tree edges) of one bi-connected component intersect with two or more non-tree
edges (or tree edges) of another bi-connected component. The correctness of this filter step is shown in Proposition \ref{cor2}.
Furthermore, an intersection is pruned if the relative order of edge weights do not change after the intersection. This is because 
in such cases, as there is no change in the relative order of edges and all the edges have distinct weights (before and after the intersection), there 
will be no change in the MST (using Lemma \ref{lem4}). 

After the filtering steps the algorithm checks if the edges involved can be directly exchanged or not. Otherwise the non edge is added to the MST. 
This addition can create only one cycle. The cycle property of minimum spanning trees \cite{CPROP} states that given a cycle, the heaviest edge in 
that cycle does not belong to any minimum spanning tree. Hence, using this we can add the non tree edges and delete the heaviest edge without creating 
any cycles or affecting the correctness of the minimum spanning tree. Thus, the EIO algorithm is correct.
\end{proof}

\noindent\textbf{Storage costs of the data structures} The edge-table has an entry for each
edge in the ST network. Thus it would have $m$ entries, where $m$ is number of
edges in the ST network. Now the number of fcycles in graph is bounded by
$O(m-n+1)$ (one for each edge deleted during construction of MST at t=1) where
$m$ and $n$ are number of the edges and nodes of ST network. Thus, the length of
the bit vector for fcycles column of edge-table is $m-n+1$ (one bit for each
fcycle). Therefore the total storage cost of edge-table is $O(m^2)$. Similarly,
the total number of entries in fcycle-table is $m-n+1$ and each entry of
fcycle-table has a list which can have a worst case length of $O(m)$. Therefore,
the total storage cost of fcycle-table is $O(m^2)$.\\

\noindent\textbf{Asymptotic analysis of EIO algorithm} The running time of the EIO algorithm is
sensitive to the number of intersection points and number of edges involved per
intersection point. Here, we consider two kinds of intersection points: one, in
which all the edges are involved and the other, where only two edges (or a
constant number) are involved. First consider the case of a two-edge
intersection. The number of two edge intersections (or constant number) between
a pair of consecutive time instants is $O(m^2)$. Since the
edge-table is indexed using hashing, all the filtering steps would take
only $O(1)$ time. Similarly, step 21 (determining $s_i$ and $s_j$) would  take
only constant time as there are only two (or constant number) edges. Steps 33-36
can take $O(n)$ ($n$ being the number of nodes in graph) time in the worst case
(when the fundamental cycle involves all the nodes of the graph). Thus, the two
edge intersection case would take $O(m^2n)$ worst case time (for one consecutive
pair of time instances in the time series). The maximum number of times this can
happen is $O(K)$ (once between every two time instances of the time series).

Now, consider the case when $O(m)$ edges intersect at a single point. In this
case step 21 would take $O(m\log m)$ time. This kind of intersection would
involve a maximum of $O(m-n+1)$ non-tree edges. Thus step 33 would take $O(n)$
per non-tree edge making a total of $O(mn + m\log m)$ time in the worst case.
This kind of intersection can happen only $O(K)$ times. This is because the edge
weight functions vary linearly between two time instances of the time series,
and thus they can all meet at only one point between two time instants of the
time series. Consider the case when length of the time series is very long. Let
the length of time interval when $O(m^2)$ intersections (two edge intersections)
occur be $K_1$ and $O(m)$ edge intersections be $K_2$ where  $K_1 + K_2 = K$.
The time required to sort the intersection points in two-edge intersection case
is bounded by $O(m^2K_1\log(m^2K_1))$, whereas it would take $O(K_2\log K_2)$
time to sort when $O(m)$ edges are involved in the intersection. Therefore the
total time required for sorting all the intersection points is
$O((m^2K_1 + K_2)\log(m^2K_1+ K_2))$. Thus the total worst case time required
by the EIO algorithm is the sum of time spent on two-edge intersections, $O(m)$
edge intersections and, the time required to sort all the intersection points.
Thus the overall time complexity of the EIO algorithm is $O(m^2nK_1 + mnK_2 +
K_2m\log m + (m^2K_1 + K_2)\log(m^2K_1+ K_2))$.\\

\begin{figure}[h]
\begin{center}
\subfigure[Comparison on time series length]{\label{exp-a}\includegraphics[width = 70mm]{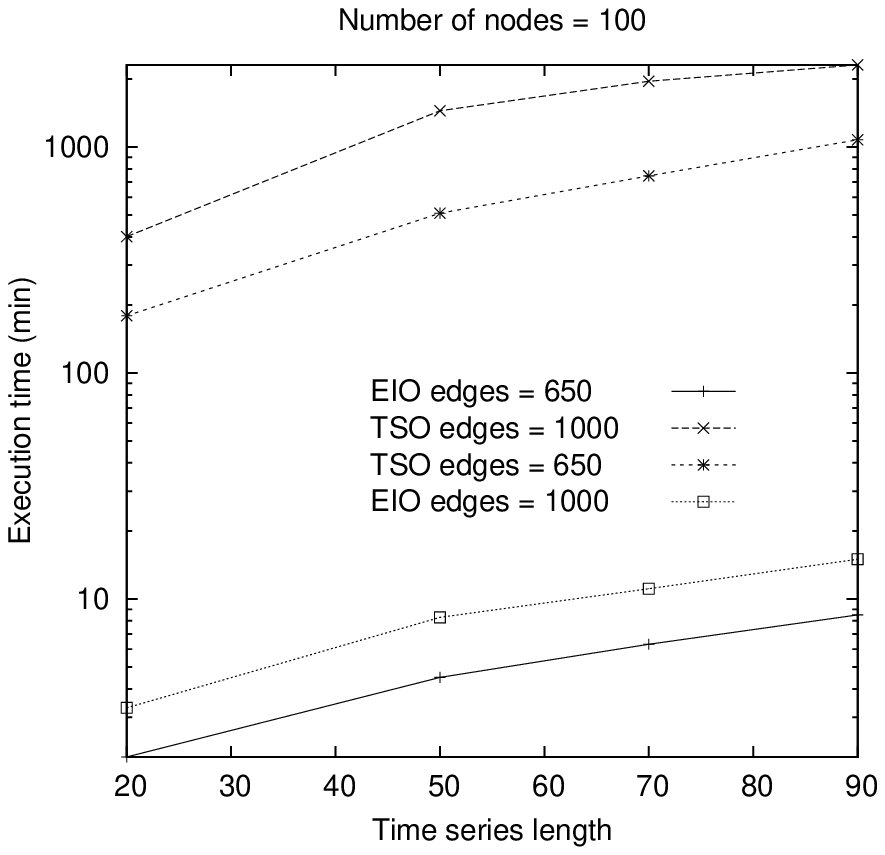}}
\subfigure[Comparison on number of edges]{\label{exp-b}\includegraphics[width = 70mm]{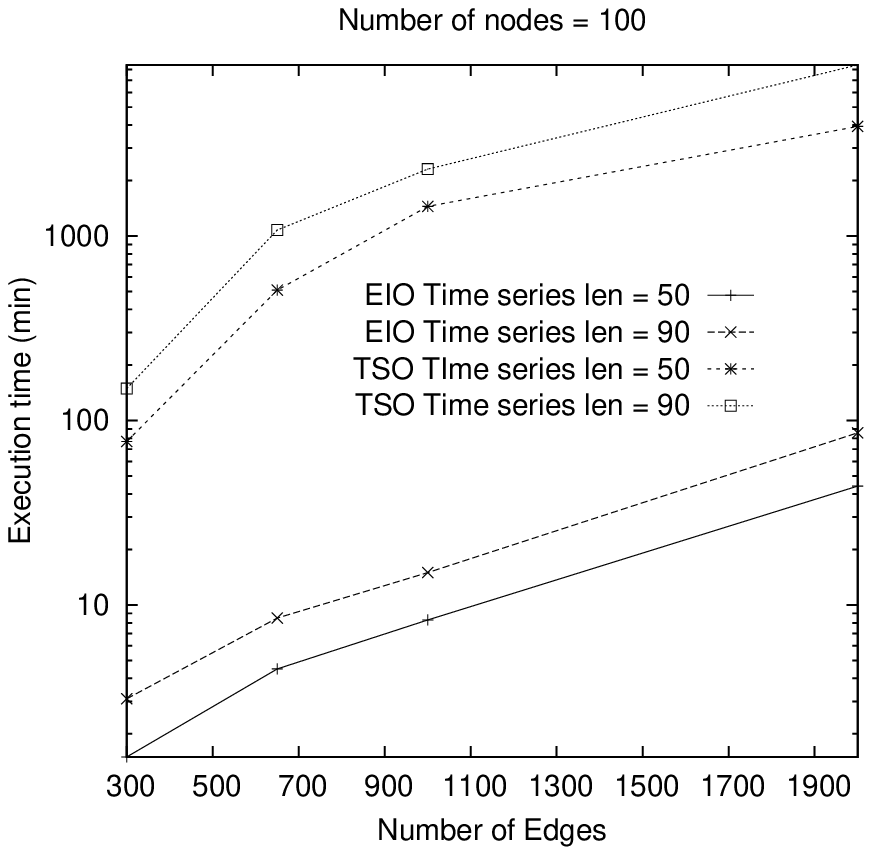}}
\caption{Comparison of EIO and TSO algorithms.}
\end{center}
\end{figure} 

\noindent\textbf{Relaxation of edge presence assumption}
Now, consider the case when the edge presence assumption is relaxed. Even though the
number of time instants in the edge weight time series increases,
the total number of intersection points do not change because we are not adding any new
weight functions. Let $L$ be the number of disjoint time intervals during which any of the edges is absent.
The bi-connected component information has to re-computed at the start
of each interval and at the end of each interval. Thus, the bi-connected component 
information has to be computed $2L$ times and each of these computations 
take $O(m+n)$ time \cite{AHO}. During each of these $L$ intervals, if any 
of the tree edges is absent then it has to replaced by a non-tree edge.
The replacement involves observing the cycles created by each of the $O(m-n+1)$ 
non-tree edges in increasing order of their weights. Since there can be at most $n-1$ tree edges, this takes $O(mn + m\log m)$ time.
Here, each tree edge is not processed individually. All the tree edges are stored in a hash table and are checked in the cycles created by the addition of each non-tree edge.
Then after an edge (either tree or non-tree) reappears, it is again added to MST and the heaviest edge is deleted from the cycle created. Now, there can be at most $m-n+1$ edges absent 
during a time interval (assuming that the graph always remains connected). Therefore, this step takes $O(mn)$ time.
Therefore the total complexity is $O(Lmn + Lm\log m)$. Thus, the overall complexity of 
the EIO algorithm becomes  $O(m^2nK_1 + mnK_2 + K_2m\log m + (m^2K_1 + K_2)\log(m^2K_1+ K_2) 
+ Lmn + Lm\log m)$. 


\section{Experimental Evaluation}
\label{exp}

\begin{figure}[h]
\begin{center}
\includegraphics[width=75mm]{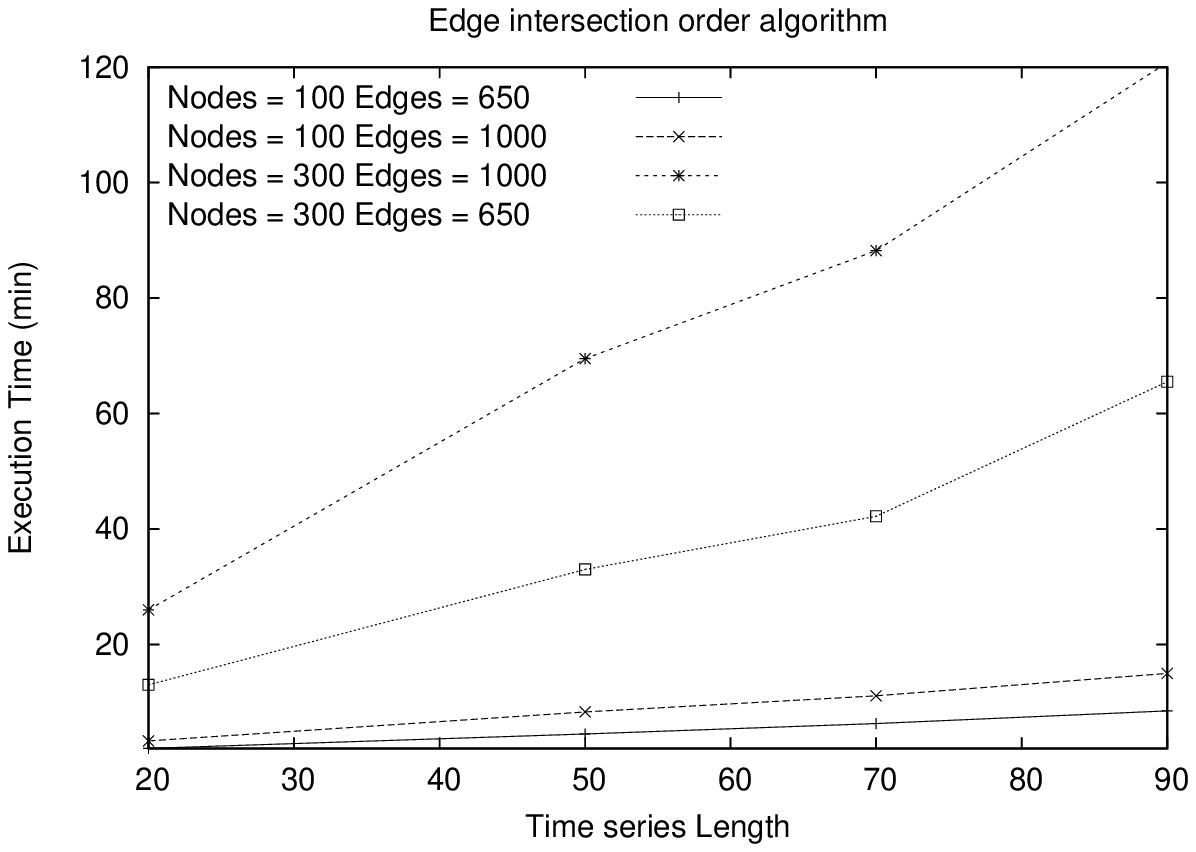}
\caption{EIO algorithm: Execution time with respect to length of time series.}
\label{plot3}
\end{center}
\end{figure}

The purpose of the experimental evaluation was to compare the execution times of
the TSO and EIO algorithms. The two algorithms were compared on synthetic
datasets. The experimental parameters that were varied in experiments are: (1)
length of time series, (2) number of edges, and (3) number of nodes. In our
experiments we generated the networks randomly. Given the number of nodes ($n$)
and number of edges ($m$) of the network, first, a spanning tree containing
$n-1$ edges is generated for the network. This is to guarantee that the network
is connected; otherwise, the TSMST cannot be determined. Edges are then randomly
added to this spanning tree till the number of edges becomes $m$. After that a
time series is associated with each edge. The time series is also generated
randomly. The experiments were conducted on an Intel Xeon workstation with
2.40GHz CPU, 8GB RAM and Linux operating system.

\subsection{Effect of Length of Time Series}

Figure~\ref{exp-a} shows the performance of EIO and TSO algorithm as the length
of the time series increases. Execution time of both the algorithms increase
with time. The figure shows a superior performance of the EIO algorithm over the
TSO algorithm. This is due to the increase of intersection points that occurs
with the increase in the length of time series. Since the TSO algorithm
recomputes the MST at each intersection point, it takes much more time
than the EIO algorithm, which just updates the MST with no recomputing.
Experiments reveal that execution time of both TSO and EIO algorithms vary
linearly with length of time time series.

\subsection{Effect of Number of Edges}

Figure~\ref{exp-b} shows the performance of EIO and TSO algorithm as the number
of edges increases. Execution time of EIO algorithm was observed to increase
quadratically. Execution time of the TSO algorithm increased much more rapidly
than that of the EIO algorithm. Both the previous experiments clearly showed the
superior performance of the EIO algorithm over the TSO algorithm. The EIO
algorithm was faster than the TSO algorithm by an order of magnitude. Moreover, the
difference in the execution times of the two algorithms increased with increase
in length of time series and number of edges.

\begin{figure}[h]
\begin{center}
\subfigure{\label{bc-a}\includegraphics[angle=270,width = 74mm]{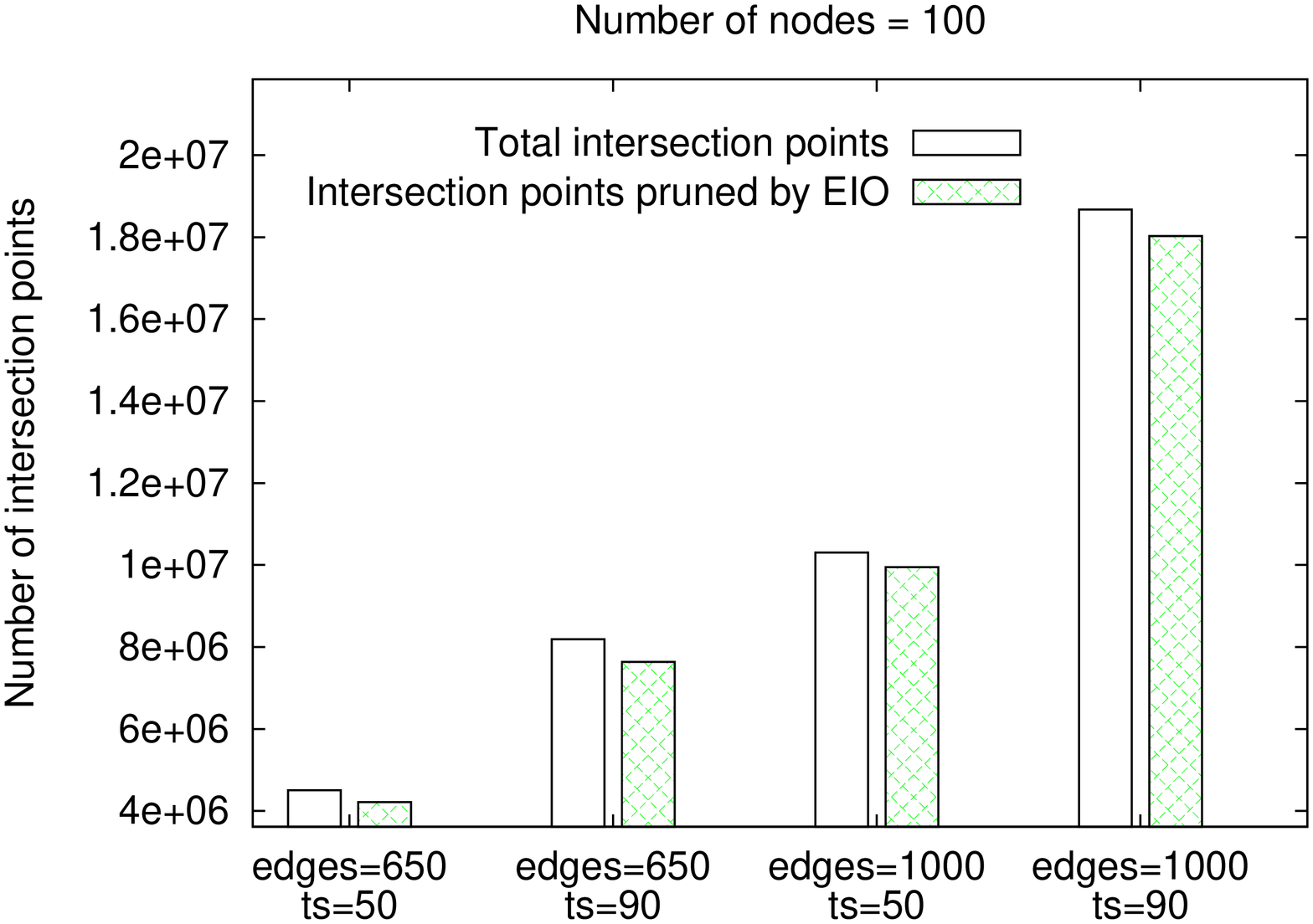}}
\subfigure{\label{bc-b}\includegraphics[angle=270,width = 74mm]{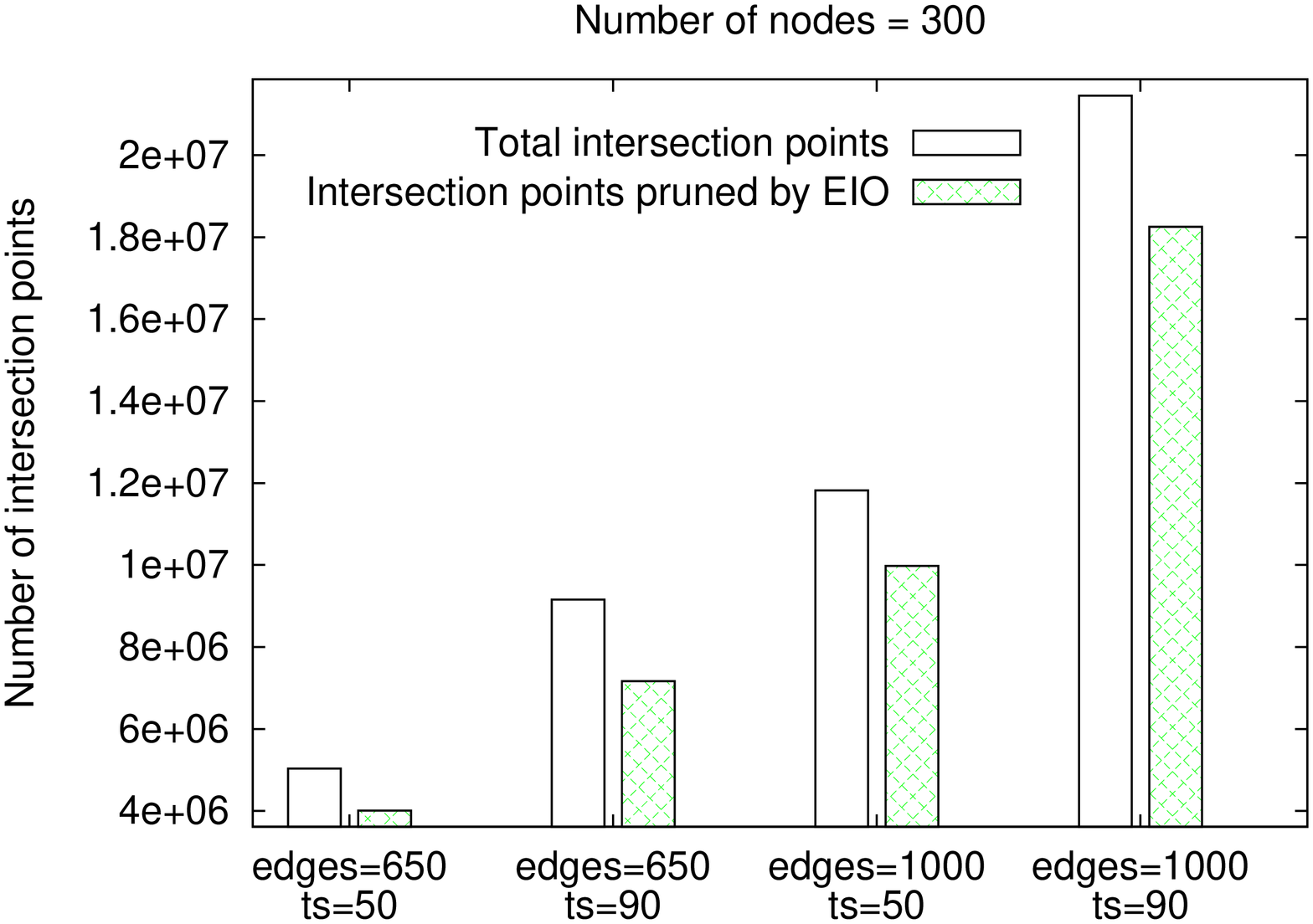}}
\caption{Number of intersection points in different datasets.}
\label{bc}
\end{center}
\end{figure}

Figure~\ref{plot3} shows the performance of the EIO algorithm for four different
network sizes. The execution time increased linearly with length of time series
in all cases. The execution time increased at a faster rate as the size of the
network increased. For instance, execution time increased much more rapidly for
a network with 300 nodes and 1000 edges than for a network with 100 nodes and
650 edges.

\subsection{Performance Evaluation of the Filters}

Figure~\ref{bc} shows the total number of the intersection points and the number
of intersection points pruned by the EIO algorithm. The figure shows that the
number of intersection points increase with increase in size of network. The
figure also shows that a large number of intersection points were pruned by the
EIO algorithm. This clearly shows the superior performance of the filters used
in the EIO algorithm. Table \ref{pf5} and Table \ref{pf9} show the percentage of 
intersection points filtered by individual filters.
     
\begin{table*}[t]
\begin{center}
\begin{tabular}{ m{2.1cm}|m{1.6cm}|m{1.6cm}|m{2.19cm}|m{2cm}|m{2.05cm}}
\cline{2-5} 
& \multicolumn{4}{c|}{Percentage of intersection points filtered} &\\
\hline
\multicolumn{1}{|m{2.1cm}|}{Network size} & Only tree edges & Only non-tree edges & Different bi-connected components & No change in relative order & \multicolumn{1}{m{2.1cm}|}{Total intersection points}\\
\hline
\multicolumn{1}{|m{2.1cm}|}{Nodes=100 Edges=130} & 73.42 & 10.69 & 0.87 & 0.07 & \multicolumn{1}{m{2.1cm}|}{356885} \\
\hline 
\multicolumn{1}{|m{2.1cm}|}{Nodes=100 Edges=150} & 58.95 & 20.06 & 0.37 & 0.05 & \multicolumn{1}{m{2.1cm}|}{481301} \\
\hline 
\multicolumn{1}{|m{2.1cm}|}{Nodes=100 Edges=650} & 5.56 & 87.74 & 0 & 0 & \multicolumn{1}{m{2.1cm}|}{8185234} \\
\hline 
\multicolumn{1}{|m{2.1cm}|}{Nodes=300 Edges=330} & 92.51 & 2.05 & 0.25 & 0.01 & \multicolumn{1}{m{2.1cm}|}{2388274} \\
\hline 
\multicolumn{1}{|m{2.1cm}|}{Nodes=300 Edges=350} & 85.77 & 4.5 & 0.57 & 0.01& \multicolumn{1}{m{2.1cm}|}{2676646} \\
\hline 
\multicolumn{1}{|m{2.1cm}|}{Nodes=300 Edges=650} & 33.7 & 44.5 & 0 & 0.01 & \multicolumn{1}{m{2.1cm}|}{9157842} \\
\hline
\end{tabular}
\caption{Performance of filters for time series length=90.}
\label{pf9} 
\end{center}
\end{table*}

\begin{table*}[t]
\begin{center}
\begin{tabular}{ m{2.1cm}|m{1.6cm}|m{1.6cm}|m{2.19cm}|m{2cm}|m{2.05cm}}
\cline{2-5} 
& \multicolumn{4}{c|}{Percentage of intersection points filtered} &\\
\hline
\multicolumn{1}{|m{2.1cm}|}{Network size} & Only tree edges & Only non-tree edges & Different bi-connected components & No change in relative order & \multicolumn{1}{m{2.1cm}|}{Total intersection points}\\
\hline
\multicolumn{1}{|m{2.1cm}|}{Nodes=100 Edges=130} & 72.88 & 10.33 & 0.9 & 0.05 & \multicolumn{1}{m{2.1cm}|}{196732} \\
\hline 
\multicolumn{1}{|m{2.1cm}|}{Nodes=100 Edges=150} & 58.79 & 19.8 & 0.49 & 0.06 & \multicolumn{1}{m{2.1cm}|}{264020} \\
\hline 
\multicolumn{1}{|m{2.1cm}|}{Nodes=100 Edges=650} & 5.72 & 87.75 & 0 & 0 & \multicolumn{1}{m{2.1cm}|}{4506094} \\
\hline 
\multicolumn{1}{|m{2.1cm}|}{Nodes=300 Edges=330} & 92.14 & 1.91 & 0.26 & 0.01 & \multicolumn{1}{m{2.1cm}|}{1317133} \\
\hline 
\multicolumn{1}{|m{2.1cm}|}{Nodes=300 Edges=350} & 85.77 & 4.27 & 0.12 & 0.02 & \multicolumn{1}{m{2.1cm}|}{1476097} \\
\hline 
\multicolumn{1}{|m{2.1cm}|}{Nodes=300 Edges=650} & 34.38 & 45.19 & 0 & 0.01 & \multicolumn{1}{m{2.1cm}|}{5031355} \\
\hline
\end{tabular}
\caption{Performance of filters for time series length=50.}
\label{pf5} 
\end{center}
\end{table*}

\section{Conclusions}
\label{conclude}
The time-sub-interval minimum spanning tree (TSMST) problem is a key component
of various spatio-temporal applications such as wireless sensor networks. 
The paper proposes two novel algorithms for TSMST
computation. The time sub-interval algorithm (TSO) computes the TSMST by
recomputing the MST at all time points where there is a possible change in the
ranking of candidate spanning trees (i.e., it recomputes the MST at all the
intersection points of edge weight functions) and then outputs the set of
distinct MSTs along with their respective time-sub-intervals. The edge
intersection order algorithm (EIO) updates the MST, only if necessary, at these
time points. Both these algorithms are based on a model for spatio-temporal
networks called time-aggregated graphs. The asymptotic complexity of the TSO algorithm was  $O(m^3K\log m + m^2K\log(m^2K))$ 
and the asymptotic complexity of the EIO algorithm was $O(m^2nK_1 + mnK_2 +
K_2m\log m + (m^2K_1 + K_2)\log(m^2K_1+ K_2))$. Computational complexity analysis shows
that the EIO algorithm is faster than the TSO by a factor of almost $O(m)$. Experiments
also show that the EIO is faster than the TSO algorithm by an order of magnitude.

In future, we plan to evaluate the performance of the algorithms using real datasets. We also
plan to extend the algorithms to give optimal solutions subject to the
constraint that the edge weight functions are non-linear in nature.

\bibliographystyle{abbrv}
\bibliography{ref}
\addcontentsline{toc}{chapter}{Bibliography}

\end{document}